\documentclass[a4paper]{article}
\usepackage{amsfonts}
\usepackage{amssymb, latexsym}
\usepackage{amsmath}
\allowdisplaybreaks
\usepackage{amsthm}

\usepackage{empheq}
        
\usepackage[
colorlinks=true,       
linkcolor=blue,          
citecolor=blue,        
plainpages=false,      
pdfpagelabels,
]{hyperref}

\usepackage{cite}

\usepackage[shortlabels]{enumitem}
\setlist[enumerate]{leftmargin=*,align=left,labelindent=\parindent}

\usepackage{graphicx}

\makeatletter
\newcommand{\myitem}[1][]{%
\item[#1]\protected@edef\@currentlabel{#1}\ignorespaces%
}
\makeatother

\newcommand{\amp}{\:\&\:}
\newcommand{\imp}{\;\rightarrow\:}
\newcommand{\defeqiv}{\stackrel{\textup{def}}{\iff}}
\newcommand{\defeql}{\stackrel{\textup{def}}{\  =\  }}
\DeclareMathOperator*{\medwedge}{{\textstyle{\bigwedge}} }
\DeclareMathOperator*{\medvee}{{\textstyle{\bigvee}}}
\newcommand{\prox}{\prec}
\newcommand{\proxop}{\succ}

\newcommand{\entails}{\mathrel{\vdash}}
\newcommand{\meets}{\between}
\newcommand{\downset}{\mathord{\downarrow}}
\newcommand{\upset}{\mathord{\uparrow}}

\newcommand{\id}{\mathrm{id}}
\newcommand{\One}{\mathbf{1}}

\newcommand{\Hom}[3]{\mathsf{Hom}(#2,#3)_{#1}}

\newcommand{\Fin}[1]{\mathsf{Fin}(#1)}
\newcommand{\PFin}[1]{\mathsf{Fin}^{+}(#1)}
\newcommand{\FFin}[1]{\mathrm{DL}(#1)}
\newcommand{\Image}[1]{\mathrm{Im}\,#1 }
\newcommand{\Choice}[1]{\mathrm{Ch}(#1)}
\newcommand{\ScottFunc}{\Sigma}
\newcommand{\Scott}[1]{\ScottFunc(\Spec{#1})}
\newcommand{\Nat}{\mathbb{N}}
\newcommand{\Rat}{\mathbb{Q}}

\newcommand{\DeGroot}[1]{{{#1}^{\mathsf{d}}}}
\newcommand{\DualLat}[1]{{{#1}^{\mathsf{\circ}}}}
\newcommand{\DualVal}[1]{{{#1}^{\mathsf{\bullet}}}}
\newcommand{\Frame}[1]{{\mathrm{\Omega} (#1})}
\newcommand{\Opposite}[1]{{{#1}^{\mathsf{op}}}}
\newcommand{\Upper}[1]{\mathrm{P_U}(#1)}
\newcommand{\UpperFunc}{\mathrm{P_U}}
\newcommand{\Lower}[1]{\mathrm{P_L}(#1)}
\newcommand{\LowerFunc}{\mathrm{P_L}}
\newcommand{\Double}[1]{\mathrm{P_D}(#1)}
\newcommand{\DoubleFunc}{\mathrm{P_D}}
\newcommand{\Vietoris}[1]{\mathrm{P_V}(#1)}
\newcommand{\VietorisFunc}{\mathrm{P_V}}
\newcommand{\Valuation}[1]{\mathfrak{V}(#1)}
\newcommand{\PValuation}[1]{\mathfrak{V_{P}}(#1)}
\newcommand{\ValuationFunc}{\mathfrak{V}}
\newcommand{\PValuationFunc}{\mathfrak{V_{P}}}
\newcommand{\CoValuation}[1]{\mathfrak{C}(#1)}
\newcommand{\PCoValuation}[1]{\mathfrak{C_{P}}(#1)}
\newcommand{\CoValuationFunc}{\mathfrak{C}}
\newcommand{\PCoValuationFunc}{\mathfrak{C_{P}}}
\newcommand{\Patch}[1]{\mathrm{Patch}(#1)}
\newcommand{\PatchFunc}{\mathrm{Patch}}
\newcommand{\PatchP}[1]{\mathrm{Patch'}(#1)}

\newcommand{\dual}[1]{{#1}^{\circ}}
\newcommand{\Space}[1]{\mathrm{Sp}(#1)}
\newcommand{\Ideals}[1]{\mathrm{Idl}(#1)}
\newcommand{\Ups}[1]{\mathrm{Upper}(#1)}
\newcommand{\RIdeals}[1]{\mathrm{RIdl}(#1)}
\newcommand{\Karoubi}[1]{\mathsf{\mathbf{Split}(#1)}}
\newcommand{\Spectral}{\mathsf{Spectral}}
\newcommand{\SpectralP}{\mathsf{Spectral_{Pre}}}
\newcommand{\DLatProx}{\mathsf{DLat_{Prox}}}
\newcommand{\DLatjProx}{\mathsf{DLat_{JProx}}}

\newcommand{\Ent}{\mathsf{Entsys}}
\newcommand{\EntRel}{\mathsf{EntRel}}

\newcommand{\SProxLat}{\mathsf{SProxLat}}
\newcommand{\SProxLatP}{\mathsf{SProxLat_{Prox}}}
\newcommand{\SProxLatPerfect}{\mathsf{SProxLat_{Perf}}}
\newcommand{\ProxLat}{\mathsf{ProxLat}}

\newcommand{\ProxLatPerfect}{\mathsf{ProxLat_{Perf}}}
\newcommand{\Cat}[1]{\mathbb{#1}}

\newcommand{\jSPxLat}{\mathsf{JSProxLat}}
\newcommand{\jSPxLatP}{\mathsf{JSProxLat_{Prox}}}

\newcommand{\ContEnt}{\mathsf{ContEnt}}
\newcommand{\SContEnt}{\mathsf{SContEnt}}
\newcommand{\SContEntP}{\mathsf{SContEnt_{Prox}}}
\newcommand{\SContEntPerfect}{\mathsf{SContEnt_{Perf}}}

\newcommand{\ContEntPerfect}{\mathsf{ContEnt_{Perf}}}

\newcommand{\EnttoLat}[1]{\mathrm{DL}(#1)}
\newcommand{\ApproxExt}[1]{\mathrel{\widetilde{#1}}}
\newcommand{\Spec}[1]{\mathrm{Spec}(#1)}

\newtheorem{theorem}{Theorem}[section]
\newtheorem{proposition}[theorem]{Proposition}
\newtheorem{lemma}[theorem]{Lemma}
\newtheorem{corollary}[theorem]{Corollary}
\theoremstyle{definition}
\newtheorem{definition}[theorem]{Definition}
\theoremstyle{remark}
\newtheorem{remark}[theorem]{Remark}
\newtheorem{notation}[theorem]{Notation}

\numberwithin{equation}{section}

\title{Presenting de Groot duality of stably compact spaces} 

\author{Tatsuji Kawai\\[.5em]
\normalsize Japan Advanced Institute of Science and Technology\\
\normalsize \texttt{tatsuji.kawai@jaist.ac.jp}}
\date{}
\begin{document}
\maketitle
\begin{abstract}
  We give a constructive account of the de Groot duality of stably
  compact spaces in the setting of strong proximity lattice, a
  point-free representation of a stably compact space. To this end, we
  introduce a notion of strong continuous entailment relation, which can be
  thought of as a presentation of a strong proximity lattice by
  generators and relations.
  The new notion allows us to identify de Groot duals of stably
  compact spaces by analysing the duals of their presentations.
  We carry out a number of constructions on strong proximity lattices
  using strong continuous entailment relations and study their de Groot
  duals. The examples include various powerlocales, patch topology,
  and the space of valuations. These examples illustrate the
  simplicity of our approach by which we can reason about the de Groot
  duality of stably compact spaces.


\end{abstract}

\section{Introduction}\label{sec:Introduction}
De Groot duality of stably compact spaces induces a family of
dualities on various powerdomain constructions. In the point-free
setting, Vickers~\cite{VickersEntailmentSystem} showed that the de
Groot dual of the upper powerlocale of a stably compact
locale
is the
lower powerlocale of its dual.%
\footnote{Stably compact locales are also known as \emph{stably locally
compact locales} \cite[Chapter VII, Section 4.6]{johnstone-82}
or \emph{arithmetic lattices} \cite{JungSunderhaufDualtyCompactOpen}.
The upper and lower powerlocales 
of a locale correspond to the Smyth and Hoare powerdomains
of the corresponding space, respectively.}
In the point-set setting,
Goubault-Larrecq~\cite{GoubaultLarrecq-ModelofChoice} showed that 
the dual of the Plotkin powerdomain of a stably compact space is the
Plotkin powerdomain of its dual; the same holds for the
probabilistic powerdomain.

In this paper, we give an alternative account of these  results in the
setting of strong proximity lattice
\cite{JungSunderhaufDualtyCompactOpen}, the
Karoubi envelop
of the category of spectral locales and locale maps. Strong proximity
lattices have a structural duality which reflects the de Groot duality
of stably compact spaces in a simple way (see Section
\ref{sec:deGrootDuality}). Moreover, a strong proximity lattice is
just a distributive lattice with an extra structure, so it does not
require infinitary joins inherent in the usual point-free approach.
This provides us with a convenient setting to study the de Groot
duality of stably compact spaces constructively.

To deal with stably compact spaces presented by generators and
relations, we introduce a notion of \emph{strong continuous entailment
relation}, which can be thought of as a presentation of a strong
proximity lattice by generators and relations by Scott's
entailment relations \cite{ScottEngenderingIllusion}.
The notion is a variant
of that of an entailment relation with the interpolation property due to
Coquand and Zhang~\cite{CoquandZhangPrediativePatch}.  Here, the
structure due to Coquand and Zhang is strengthened so that it has an
intrinsic duality which reflects the de Groot duality of stably
compact spaces.  The resulting
structure, strong continuous entailment relation, allows us to identify
de Groot duals of stably compact locales presented by generators
and relations by analysing the duals of their presentations.
We illustrate the ease with which we can reason about de Groot
duality by carrying out a number of constructions on strong proximity
lattices using strong continuous entailment relations. The examples include
various powerlocales, patch topology, and the space
of valuations.

Throughout this paper, we work in the point-free setting,
identifying stably compact spaces with their point-free counterparts,
stably compact locales.
This allows us to work constructively in the predicative sense as
manifested in Aczel's constructive set theory
\cite{Aczel-Rathjen-Note}.
However, the point of this work is not the constructively but
the simplicity of our approach by which we can analyse de Groot duals of
various constructions on stably compact spaces.

\subsubsection*{Related works}
Besides the work of Coquand and Zhang~\cite{CoquandZhangPrediativePatch}
and that of Jung
and S\"underhauf \cite{JungSunderhaufDualtyCompactOpen} mentioned
above, many authors studied stably compact spaces from the point-free
perspective (see Escard\'o~\cite{escardo2001regular}; Jung, Kegelmann,
and Moshier~\cite{JungEtAlMultilingual}; Vickers
\cite{VickersEntailmentSystem}).
Among them, the notion of entailment system by Vickers
\cite{VickersEntailmentSystem}, which develops the idea of Jung et
al.~\cite{JungEtAlMultilingual}, is particularly related to the notion
of strong continuous entailment relation. These structures are
equipped with structural dualities which reflects the de Groot duality
of stably compact spaces.

The essential difference between our approach and that of Vickers is
the following: 
the theory of strong continuous entailment relation is built on the
fact that stably compact locales are the retracts of spectral locales
and locale maps. Hence, the theory of strong continuous entailment
relation essentially deals with the objects of the
Karoubi envelop of the latter category.
On the other hand, the theory of entailment system deals with the
objects of the Karoubi envelop of the category of spectral locales and
preframe homomorphisms (see Section \ref{sec:ContEnt}).%
\footnote{More specifically, it suffices to consider only \emph{free
frames} rather than spectral locales.\label{foot:FreeFrame}} 
In this view, the former theory treats stably compact locales as
locales while the latter theory treats them as preframes; this has to
do with the simplicity of the treatment of joins in the geometric presentation of a
locale represented by the former structure (see Section
\ref{sec:StrongProximityLattice}). Thus, if one is interested in the
localic structure of stably compact locales rather than that of
preframe, it would be more natural to work with strong continuous
entailment relations. In particular, this could potentially facilitate
some of the localic constructions on stably compact locales involving
finite joins, such as patch topology and the Vietoris powerlocale, in the
setting of strong continuous entailment relations, although proper
comparison is needed.

Apart from the point-free approaches mentioned above, we are motivated
by the corresponding results for stably compact spaces due to
Goubault-Larrecq~\cite{GoubaultLarrecq-ModelofChoice}. To derive these
results, he used the notion of $\mathbf{A}$-valuation due to Heckmann
\cite{HECKMANN1997160}. It would be interesting to know if there is
any connection between our approach and the $\mathbf{A}$-valuation
approach. However, since we prefer to work constructively in the
point-free setting, we do not compare the two approaches in this paper.

\subsubsection*{Organisation}
In Section \ref{sec:PreliminaryLocale}, we fix some basic notions on locales.
In Section \ref{subsec:ProximityLattice}, we introduce 
the notion of proximity lattice as the Karoubi envelop of the
category of spectral locales and preframe homomorphisms.
In Section \ref{sec:ContEnt}, we give an alternative representation
for proximity lattices, called continuous entailment relation,
based on the notion of entailment system.
In Section~\ref{sec:StrongProximityLattice}, we strengthen the notion of proximity
lattice to strong proximity lattice by looking into the Karoubi
envelop of the category of spectral locales and locale maps. We also
introduce the corresponding notion of strong continuous entailment
relation.
In Section~\ref{sec:deGrootDuality}, we formulate the duality of 
proximity lattices and continuous entailment relations,
and show that these dualities reflect the de Groot duality
of stably compact locales.
In Section \ref{sec:ExampledeGrootDuality}, 
we study the de Groot duals of various constructions on stably compact
locales by exploiting the correspondence between strong proximity lattices
and strong continuous entailment relations.

\section{Preliminary on locales}\label{sec:PreliminaryLocale}
A \emph{frame} is a poset $(X, \wedge, \bigvee)$ with 
finite meets $\wedge$ and joins $\bigvee$ for all subsets of
$X$ where finite meets distribute over all joins.
A homomorphism from a frame $X$ to a frame $Y$ is a function $f \colon
X \to Y$ which preserves finite meets and all joins.  The
\emph{category of locales} is the opposite of the category of frames and
frame homomorphisms. We write $\Frame{X}$ for the
frame corresponding to a locale $X$, but we often regard a frame as
a locale and vice versa without change of notation.

Given a set $S$ of \emph{generators}, a \emph{geometric theory} over $S$ is a set of
\emph{axioms} of the form
  $
  \medwedge A \entails \bigvee_{i \in I} \medwedge  B_i,
  $
where $A$ is a finite subset of 
$S$ and $(B_{i})_{i \in I}$ is a set-indexed family of finite subsets
of $S$.%
\footnote{Here, \emph{finite} means \emph{finitely enumerable}.
A set $A$ is finitely enumerable if there is  a surjective
function $f : \left\{ 0,\dots,n-1 \right\} \to A$ for some
natural number $n$.
Finitely enumerable sets are also known as Kuratowski finite sets; see
e.g., Johnstone~\cite[D5.4]{ElephantII}.}
Single conjunctions and single disjunctions are identified with elements
of $S$.  We use the following abbreviations:
\begin{align*}
  \top &\equiv \medwedge \emptyset, & \bot &\equiv \bigvee \emptyset, &
  \medvee B &\equiv \bigvee_{b \in B} \left\{ b \right\}, &
  \bigvee_{i \in I} b_{i} &\equiv \bigvee_{i \in I} \left\{ b_{i} \right\}.
\end{align*}
An \emph{interpretation}
of a geometric theory $T$ (over $S$) in a locale $X$ is a function $f \colon S \to \Frame{X}$ such that
  $
  \medwedge_{a \in A}f(a) \leq_{X} \bigvee_{i \in I}
  \medwedge_{b \in B_{i}} f(b)
  $
for each axiom $\medwedge A \entails \bigvee_{i \in I} \medwedge B_i$ of
$T$. There is a locale $\Space{T}$ with a universal
interpretation $i_{T} \colon S \to \Frame{\Space{T}}$: for any
interpretation $f \colon S \to \Frame{X}$ of $T$, there exists a unique frame
homomorphism $\overline{f} \colon \Frame{\Space{T}} \to \Frame{X}$
such that $\overline{f} \circ i_{T} = f$.
In this case, $\Space{T}$
is called the \emph{locale (or frame) presented by
$T$}. 
A \emph{model} of a geometric theory $T$ over $S$ is a subset $\alpha
\subseteq S$ such that
  $
  A \subseteq \alpha \imp  \exists i \in I
  \left( B_i \subseteq \alpha \right)
  $
for each axiom $ \medwedge A \entails \bigvee_{i \in I}  \medwedge  B_i$ of $T$.
If the models of $T$ form a distinguished class of
objects, we call $\Space{T}$  the \emph{locale
whose models are members of that class}.%
\footnote{
    We differ from the standard distinction between ``interpretation''
    and ``model'' wherein the former interprets the language (here
    generators) and the latter in addition satisfies the
    axioms of the theory over the language.
    In this paper, in contrast, both ``interpretation'' and ``model''
    mean an interpretation which satisfies the axioms: the latter
    keeps the usual meaning of a model in the lattice of truth values,
    the powerset of a singleton $\left\{ * \right\}$, whereas the
    former means an axiom preserving interpretation in a frame more
    general than that of the truth values. In locale theory, this can
    be expressed as the distinction between \emph{generalised points} and
    \emph{global points} (see Vickers~\cite{DoublePowLocExp}).
}

\section{Proximity lattices} \label{subsec:ProximityLattice}
We recall the construction of Karoubi envelop (cf.\ \cite[Chapter 2,
Exercise B]{FreydAbelCat}).
\begin{definition}
  An \emph{idempotent} in a category $\Cat{C}$ is a morphism $f \colon
  A \to A$  such that $f \circ f = f$.  The \emph{Karoubi envelop} (or
  \emph{splitting of idempotents}) of $\Cat{C}$ is a category
  $\Karoubi{\Cat{C}}$ where objects are idempotents in $\Cat{C}$ and
  morphisms $h \colon (f \colon A \to A) \to (g \colon B \to B)$
  are morphisms $h \colon A \to B$ in $\Cat{C}$ such that $g \circ h =
  h = h \circ f$. 
\end{definition}
One can show that if $\Cat{C}$ is a full subcategory of 
$\Cat{D}$ where every idempotent splits in $\Cat{D}$ and every
object in $\Cat{D}$ is a retract of an object of $\Cat{C}$, then
$\Cat{D}$ is equivalent to $\Karoubi{\Cat{C}}$.

It is well known that stably compact locales are exactly the
retracts of spectral locales, whose frames are the ideal
completions of distributive lattices~\cite[Chapter VII, Theorem
4.6]{johnstone-82}. 
Less well known is the fact that stably compact locales are
exactly the preframe retracts of spectral locales~\cite[Section
3]{VickersEntailmentSystem}%
\footnote{
  To be precise, the results of Vickers~\cite{{VickersEntailmentSystem}} are
  stronger; the stably compact locales are preframe retracts of free 
  frames (cf.\ footnote~\ref{foot:FreeFrame}). }
so that the category of stably compact
locales and preframe homomorphisms can be characterised as the 
Karoubi envelop of the category of spectral locales and preframe
homomorphisms.
Here, a \emph{preframe} is a poset with directed joins (joins of
directed subsets) and finite meets which distribute over directed joins.
A \emph{preframe homomorphism} between preframes is a function which
preserves finite meets and directed joins.
The latter fact leads to the notion of proximity
lattice~\cite{JungSunderhaufDualtyCompactOpen} by considering a
finitary description of the dual of the category of spectral locales
and preframe homomorphisms.

\begin{definition}
  Let $S$ and $S'$ be
  distributive lattices. A \emph{proximity relation} 
  from $S$ to $S'$ is a relation $r \subseteq S \times S'$ such that 
  \begin{enumerate}
      \myitem[(ProxI)]\label{def:ProximityRelationVee}
      $r^{-} b \defeql \left\{ a \in S \mid a \mathrel{r} b \right\}$
      is an ideal of $S$ for all $b \in S'$,
    \myitem[(ProxF)]\label{def:ProximityRelationWedge} 
      $r a \defeql \left\{ b \in S' \mid a \mathrel{r} b \right\}$
      is a filter of $S'$ for all $a \in S$.
  \end{enumerate}
Here, an \emph{ideal} is a downward closed subset of $S$ closed under finite
joints. A \emph{filter} is an upward closed subset of $S$ closed under finite meets.
\end{definition}
Let $\DLatProx$ be the category of distributive lattices and proximity
relations: the identity on a distributive lattice $S$ is the order
$\leq$ on $S$; the composition of proximity relations is the
relational composition.

The \emph{ideal completion} of a distributive lattice $S$,
denoted by $\Ideals{S}$, is the frame of ideals of $S$:
the directed join of ideals is their union; finite joins and
finite meets are defined by
\begin{align*}
  0 &\defeql \left\{ 0 \right\},
  &
  I \vee J &\defeql \bigcup_{a \in I, b \in J} \downset \left( a \vee
  b\right),\\
  1 &\defeql S, 
  &
  I \wedge J &\defeql \left\{ a \wedge b \mid a \in I, b \in J \right\},
\end{align*}
where
  $
  \downset a \defeql \left\{ b \in S \mid b \leq a \right\},
  $
the \emph{principal ideal}
generated by $a$.
Every ideal $I$ is a directed join of principal
ideals: 
\begin{equation}
  \label{eq:PrincipleGenerate}
  I = \bigvee_{a \in I} \downset a.
\end{equation}

\begin{proposition}
  \label{prop:ProximityRelPreframeHom}
  For any proximity relation $r \colon S \to S'$, there exists
  a unique preframe homomorphism $f \colon \Ideals{S'} \to \Ideals{S}$
  such that $f(\downset b) = r^{-}b$ for all $b \in S'$.
  Moreover, this bijection preserves identities and
  compositions of proximity relations.
\end{proposition}
\begin{proof}
  It is easy to see that a proximity relation
  $r \colon S \to S'$ uniquely extends to a meet-semilattice
  homomorphism
  $f_{r} \colon S' \to \Ideals{S}$ defined by
  \begin{equation}
  \label{eq:ProximityRelMeetHom}
    f_{r}(b) \defeql r^{-}b.
  \end{equation}
  By Vickers \cite[Theorem 9.1.5 (i) (iv)]{vickers1989topology},
  the function $f_{r}$ extends uniquely  to a preframe homomorphism
  $f \colon \Ideals{S'} \to \Ideals{S}$ by 
  \begin{equation}
  \label{eq:ProximityRelPreframeHom}
  f(I) \defeql \bigvee_{b \in I} f_{r}(b) = r^{-} I.
  \end{equation}
  The second statement follows from the first 
  and \eqref{eq:PrincipleGenerate}.
\end{proof}

Let $\SpectralP$ be the category of spectral locales and preframe
homomorphisms: objects of $\SpectralP$ are distributive lattices
and morphisms are preframe homomorphisms between the ideal
completions.
\begin{theorem}
  \label{thm:DLatProxEquivSpectralP}
  The category $\DLatProx$ is dually equivalent to $\SpectralP$.
\end{theorem}
\begin{proof}
Immediate from Proposition \ref{prop:ProximityRelPreframeHom}.
\end{proof}
Since $\Karoubi{\SpectralP}$ is equivalent
to the category of stably compact locales and preframe homomorphisms,
$\Karoubi{\DLatProx}$ is dually equivalent to the latter category.
The objects and morphisms of $\Karoubi{\DLatProx}$ are called
\emph{proximity lattices} and \emph{proximity relations} respectively
(cf.\ Jung and S\"underhauf \cite{JungSunderhaufDualtyCompactOpen} and
de Gool \cite{vanGoolDualityCanExt}).  In what follows, we write $\ProxLat$ for
$\Karoubi{\DLatProx}$.

\begin{notation}
  \label{not:ProxLat}
  We write $(S, \prox)$ for a proximity lattice, where $S$ is a
  distributive lattice and $\prox$ is an idempotent proximity relation
  on $S$. We write $r \colon (S, \prox) \to (S', \prox')$ for
  a proximity relation from $(S, \prox)$ to $(S', \prox')$, i.e.,
  a proximity relation $r \colon S \to S'$ between the underlying
  distributive lattices such that ${\prox'} \circ r =
  r \circ {\prox} = r$.
\end{notation}

Each proximity lattice $(S, \prox)$ represents a stably compact locale
whose frame is the collection $\RIdeals{S}$ of \emph{rounded ideals} of $S$ ordered
by inclusion~\cite[Theorem~11]{JungSunderhaufDualtyCompactOpen}. Here,
an ideal $I \subseteq S$ is \emph{rounded} if 
$a \in I \leftrightarrow \exists b \succ a \left( b \in I \right)$.%
\footnote{The notion of rounded ideal makes sense 
in the more general setting of information systems~\cite{Infosys}.}
Directed joins and finite meets
in $\RIdeals{S}$ are
calculated as in $\Ideals{S}$; on the other
hand, finite joins in $\RIdeals{S}$ are defined by
\begin{align*}
  0 &\defeql  \downset_{\prox} 0,
  &
  I \vee J &\defeql \bigcup_{a \in I, b \in J} \downset_{\prox} \left(
  a \vee b\right),
\end{align*}
where $\downset_{\prox} a \defeql \left\{ b \in S \mid b \prox a \right\}$.
Every rounded ideal $I$ is a directed join of its members: $I = \bigvee_{a \in I} \downset_{\prox} a$.
Let $\Spec{S}$ denote the locale whose frame is $\RIdeals{S}$.

\section{Continuous entailment relations}\label{sec:ContEnt}
We give an alternative presentation of a proximity lattice in terms of
Vickers's entailment system~\cite{VickersEntailmentSystem}.
We need some constructions on finite subsets.
For any set $S$, let $\Fin{S}$ denote the set of finite subsets of
$S$.
For each $\mathcal{U} \in \Fin{\Fin{S}}$, define $\mathcal{U}^{*} \in
\Fin{\Fin{S}}$ inductively by
\begin{align*}
  \emptyset^{*}
  &\defeql \left\{ \emptyset \right\},&
  \left( \mathcal{U} \cup \left\{ A \right\} \right)^{*}
  &\defeql
  \left\{ B \cup C \mid B \in \mathcal{U}^{*} \amp C \in \PFin{A}\right\},
\end{align*}
where $\PFin{A}$ denotes the set of inhabited finite subsets of $A$.%
\footnote{
In the notation of Vickers~\cite[Section 4]{VickersEntailmentSystem},
the set $\mathcal{U}^{*}$ is equal to
  $
  \left\{ \Image{\gamma} \mid \gamma \in \Choice{\mathcal{U}}\right\},
  $
where $\Choice{\mathcal{U}}$ is the set of choices of $\mathcal{U}$
and $\Image{\gamma}$ is the image of a choice $\gamma$; see Definition
12 and Definition 13, and the proof of Proposition~14 in
\cite{VickersEntailmentSystem}.}
Writing $\FFin{S}$ for the free distributive lattice over $S$, the
mapping $\mathcal{U} \mapsto \mathcal{U}^{*}$ transforms a disjunction
of conjunctions of generators in $\FFin{S}$ to the
conjunction of disjunctions of generators, or the other way around
(cf.\ Vickers~\cite[Theorem 8.7]{DoublePowLocExp}).%

\begin{definition}
  Let $S,S',S''$ be sets and $r,s$ be relations 
  $r \subseteq \Fin{S} \times \Fin{S'}$ and $s \subseteq
  \Fin{S'} \times \Fin{S''}$. 
  \begin{enumerate}
    \item The relation $r$ is said to be \emph{upper} if 
    $
    A \mathrel{r} B \imp A \cup A' \mathrel{r} B \cup B'
    $
  for all $A,A' \in \Fin{S}$ and $B, B' 
  \in \Fin{S'}$.
    \item 
  The \emph{cut composition} $s \cdot r
  \subseteq \Fin{S} \times \Fin{S''}$ is defined by
  \[
    A \mathrel{s \cdot r} C \defeqiv \exists \mathcal{V} \in
    \Fin{\Fin{S'}} \left[ \forall B' \in \mathcal{V}^{*}
      \left( A \mathrel{r} B' \right)  \amp \forall B \in \mathcal{V} \left( 
    B \mathrel{s} C\right) \right].%
  \footnote{In \cite{VickersEntailmentSystem}, the cut
  composition of $r$ and $s$ is denoted by $r \dagger s$ using
  the forward notation for the relational composition. See in particular
  \cite[Lemma 30]{VickersEntailmentSystem}.}
  \]
  \end{enumerate}
\end{definition}

\begin{definition}[{Vickers~\cite[Section 6]{VickersEntailmentSystem}}]
  An \emph{entailment system} is a pair $(S, \ll)$ where $S$ is a set
  and $\ll$ is an upper relation on $\Fin{S}$ such that $\ll \cdot \ll
  {=} \ll$.
  A \emph{Karoubi morphism} $r \colon (S, \ll) \to (S', \ll')$ of
  entailment systems is an upper relation $r \subseteq \Fin{S} \times
  \Fin{S'}$ such that
    $
    \ll' \mathop{\cdot} r = r = r \mathop{\cdot} \ll\!.
    $
\end{definition}
Let $\Ent$ be the category of entailment systems and Karoubi  morphisms between
them: the identity on $(S,\ll)$ is $\ll$  and the composition of morphisms
is the cut composition; see Vickers~\cite[Section 5 and Section 6]{VickersEntailmentSystem} for the details.
As shown in \cite{VickersEntailmentSystem}, $\Ent$ is
dually
equivalent to the category of stably compact locales and preframe
homomorphisms. An entailment system $(S, \ll)$ represents a stably
compact locale (or frame) which is a preframe retract of the free frame over $S$, or
equivalently, the spectral locale determined by the free
distributive lattice over $S$.  The relation $\ll$ then represents the
retraction.

We now focus on the full subcategory of $\Ent$
consisting of \emph{reflexive} entailment systems~\cite[Section
6.1]{VickersEntailmentSystem}, also known as \emph{entailment
relations}~\cite{ScottEngenderingIllusion, cederquist2000entailment}.

\begin{definition}
  \label{def:EntRel}
An \emph{entailment relation} on a set $S$ is a relation
$\entails$ on $\Fin{S}$ such that
\[
  a \entails a \; (\mathrm{R})
  \qquad
  \qquad
  \frac{A \entails B}{A',A \entails B,B'} \; (\mathrm{M})
  \qquad
  \qquad
  \frac{A \entails B,a \quad a,A \entails B}{A \entails B} \; (\mathrm{T})
\]
for all $a \in S$ and $A,A',B,B' \in \Fin{S}$, where ``$a$'' denotes
$\left\{ a \right\}$ and ``$,$'' denotes the union.
\end{definition}
Each entailment relation $(S, \entails)$ is an entailment system; we
write $\EntRel$ for the full subcategory of $\Ent$ consisting of
entailment relations. As noted by Vickers \cite[Section
6.1]{VickersEntailmentSystem},
the category $\EntRel$ is dually equivalent to $\SpectralP$. Hence
$\EntRel$ is equivalent to $\DLatProx$, which we now elaborate.

\begin{definition}
  Let $S,S'$ be sets and $r$ be a relation $r \subseteq \Fin{S}
  \times \Fin{S'}$.  Define a relation ${\ApproxExt{r}} \subseteq
  \Fin{\Fin{S}} \times \Fin{\Fin{S'}}$ by 
  \[
    \mathcal{U} \mathrel{\ApproxExt{r}} \mathcal{V}
    \defeqiv
    \forall A \in \mathcal{U}
    \forall B \in \mathcal{V}^{*}
    \left( A \mathrel{r} B \right).%
    \footnote{In Vickers's notation \cite[Proposition
    25]{VickersEntailmentSystem}, we have $\mathcal{U}
    \mathrel{\ApproxExt{r}} \mathcal{V} \iff \mathcal{U}
    \mathrel{\overline{r}} \mathcal{V}^{*}$.}
  \]
\end{definition}
\begin{lemma}
  \label{lem:ApproxExtComposition}
  Let $S,S',S''$ be sets and $r,s$ be upper relations 
  $r \subseteq \Fin{S} \times \Fin{S'}$ and $s \subseteq
  \Fin{S'} \times \Fin{S''}$. Then
    $
      \ApproxExt{s \cdot r} {=} \ApproxExt{s} \circ \ApproxExt{r}.
    $
\end{lemma}
\begin{proof}
  See Vickers \cite[Proposition 22 and Proposition 31]{VickersEntailmentSystem}.
\end{proof}

An entailment relation $(S, \entails)$ determines a distributive
lattice $\EnttoLat{S,\entails}$ whose underlying set is
$\Fin{\Fin{S}}$ equipped with an equality defined by
\begin{equation*}
  \mathcal{U}  =_{\entails} \mathcal{V} 
  \defeqiv
  \mathcal{U} \mathrel{\ApproxExt{\entails}} \mathcal{V} \amp
  \mathcal{V} \mathrel{\ApproxExt{\entails}} \mathcal{U}.
\end{equation*}
The lattice structure is defined as follows:
\begin{align*}
  0 &\defeql \emptyset, &
  \mathcal{U} \vee \mathcal{V} &\defeql \mathcal{U} \cup
  \mathcal{V},\\
  %
  1 &\defeql \left\{ \emptyset \right\}, &
  \mathcal{U} \wedge \mathcal{V} &\defeql \left\{ A \cup B \mid A \in \mathcal{U} \amp B
  \in \mathcal{V} \right\}.
\end{align*}
It is easy to check that joins and meets are well-defined with
respect to $=_{\entails}$ and that the order
determined by $\EnttoLat{S,\entails}$ is $\ApproxExt{\entails}$.

If $r \colon (S, \entails) \to (S', \entails')$
is a Karoubi morphism, then
\[
  \ApproxExt{\entails'} \circ
  \ApproxExt{r}
  {=} \ApproxExt{{\entails'} \cdot r}
  {=} \ApproxExt{r}
  {=} \ApproxExt{r \cdot {\entails}}
  {=} \ApproxExt{r} \circ \ApproxExt{\entails}
\]
by Lemma \ref{lem:ApproxExtComposition}.
Then, it easy to see that $\ApproxExt{r}$  is a proximity relation 
from $\EnttoLat{S, \entails}$ to $\EnttoLat{S', \entails'}$.

\begin{proposition}
  \label{prop:EntRelEquivDLatProx}
  The assignment $r \mapsto {\ApproxExt{r}}$ determines a functor $E
  \colon \EntRel \to \DLatProx$, which establishes equivalence of
  $\EntRel$ and $\DLatProx$.
\end{proposition}
\begin{proof}
  Functoriality of $E$ follows from the construction of
  $\EnttoLat{S, \entails}$ and Lemma~\ref{lem:ApproxExtComposition}.
  Moreover, $E$ is faithful because $A \mathrel{r} B 
  \leftrightarrow \left\{ A \right\} \mathrel{\ApproxExt{r}} \left\{ B
  \right\}^{*}$.
  To see that $E$ is full, 
  for any proximity relation
  $r \colon \EnttoLat{S, \entails} \to
  \EnttoLat{S', \entails'}$, define
  $\hat{r} \subseteq \Fin{S} \times \Fin{S'}$ by
    \begin{align*}
    A \mathrel{\hat{r}} B 
    &\defeqiv
    \left\{ A \right\} \mathrel{r} \left\{ B \right\}^{*}.\\
  \intertext{
  Clearly, we have $\mathcal{U} \mathrel{\ApproxExt{\hat{r}}} \mathcal{V}
  \leftrightarrow
    \mathcal{U} \mathrel{r} \mathcal{V}$. Moreover}
      A \mathrel{(\entails' \cdot \hat{r})} B
      &\iff
      \left\{ A \right\} \mathrel{(\ApproxExt{{\entails'} \cdot \hat{r}})} \left\{
      B \right\}^{*} \\
      &\iff
      \left\{ A \right\} \mathrel{({\ApproxExt{\entails'}} \circ
      {\ApproxExt{\hat{r}}})} \left\{
      B \right\}^{*} \\
      &\iff
      \left\{ A \right\} \mathrel{({\ApproxExt{\entails'}} \circ r)}
      \left\{ B \right\}^{*} \\
      &\iff
      \left\{ A \right\} \mathrel{r} \left\{ B \right\}^{*} \\
      &\iff
      A \mathrel{\hat{r}} B.
  \end{align*}
  Similarly $\hat{r} \cdot {\entails} = \hat{r}$,  so $\hat{r}$ is a Karoubi morphism.
    Thus $E$ is full.
    To see that $E$ is essentially surjective, 
    for any distributive lattice $(S, 0, \vee,
  1, \wedge)$, define an
  entailment relation $(S, \entails)$ by
  \begin{equation}
    \label{eq:EntofDLat}
    A \entails B \defeqiv \medwedge A \leq \medvee B.
  \end{equation}
  Then, define relations
  $r \subseteq S \times \Fin{\Fin{S}}$ 
  and 
  $s \subseteq  \Fin{\Fin{S}} \times S$ by
  \begin{align*}
    a \mathrel{r} \mathcal{U} 
    &\defeqiv
    a \leq \bigvee_{A \in \mathcal{U}} \medwedge A,  &
   \mathcal{U} \mathrel{s} a
    &\defeqiv
    \bigvee_{A \in \mathcal{U}} \medwedge A \leq a.
  \end{align*}
  It is straightforward to show that $r$ and $s$ 
  are proximity relations between $S$ and $\EnttoLat{S, \entails}$ and
  inverse to each other.
\end{proof}
\begin{remark}
  \label{rem:QuasiInverseofE}
  By the standard construction (Mac Lane
  \cite[Chapter~IV, Section~4, Theorem~1]{Category_at_work}),
  we can define a quasi-inverse $D \colon \DLatProx \to \EntRel$
  of $E \colon \EntRel \to \DLatProx$ by
  \[
    D(S,0, \vee, 1, \wedge) \defeql (S, \entails),
  \]
  where $\entails$ is defined by \eqref{eq:EntofDLat}. 
  The functor $D$ sends a proximity
  relation $r \colon S \to S'$ to a Karoubi morphism
  $D(r) \colon D(S) \to D(S')$ defined by
  \[
    A \mathrel{D(r)} B
    \defeqiv \medwedge A \mathrel{r} \medvee
    B.
  \]
\end{remark}

By Proposition \ref{prop:EntRelEquivDLatProx},
the category $\Karoubi{EntRel}$ is equivalent to $\ProxLat$, and hence
 dually equivalent to the category of stably compact locales and
preframe homomorphisms. 
Unfolding the definition  of $\Karoubi{\EntRel}$, we have the
following characterisations of its objects and morphisms (Definition
\ref{def:ContEntRel} and \ref{def:ProxMap}).

\begin{definition}
  \label{def:ContEntRel}
  A \emph{continuous entailment relation} is an entailment relation
  $(S, \entails)$ equipped with an idempotent Karoubi endomorphism 
  $\ll \mathop{\subseteq} \Fin{S} \times \Fin{S}$ on $(S, \entails)$.
  We write $(S, \entails, \ll)$ for a continuous entailment relation.
\end{definition}
Note that each continuous entailment relation $(S, \entails, \ll)$ 
has an associated entailment system $(S, \ll)$.
\begin{definition}
  \label{def:ProxMap}
  Let $(S, \entails, \ll)$ and $(S', \entails', \ll')$
  be continuous entailment relations.
  A \emph{proximity map} from 
  $(S, \entails, \ll)$ to 
  $(S', \entails', \ll')$ is a Karoubi morphism between entailment systems
  $(S, \ll)$ and $(S', \ll')$.
\end{definition}
In the following, we write $\ContEnt$ for $\Karoubi{\EntRel}$.
Since an entailment system $(S, \ll)$ can be identified with a continuous
entailment relation $(S, \meets, \ll)$ where
\[
  A \meets B \defeqiv \left( \exists a \in S \right) a \in A \cap B,
\]
the category $\Ent$ can be regarded as a full
subcategory of $\ContEnt$.  On the other hand, each continuous
entailment relation $(S, \entails, \ll)$ is isomorphic to $(S, \meets,
\ll)$ with $\ll$ being the isomorphism.
Thus, we have the following.
\begin{proposition}
  \label{prop:ContEntEquivEnt}
  The categories $\ContEnt$ and $\Ent$ are equivalent.
  \qed
\end{proposition}

\section{Strong proximity lattices} \label{sec:StrongProximityLattice}
Proximity lattices and continuous entailment relations have nice
structural duality, which will be elaborated in Section
\ref{sec:deGrootDuality}. However, stably compact locales
represented by these structures do not seem to admit 
simple geometric presentations. In the case of a proximity lattice $(S,
\prox)$, for example, $\Spec{S}$ can be presented by
a geometric theory over $S$ with the following axioms:
\begin{equation}
  \label{eq:GeoTheoryProxLat}
  \begin{gathered}
    a \entails \bot
    \quad (\text{if $a \prox 0$}),
   \quad \qquad
    c \entails a \vee b
    \qquad (\text{if $c \prox a' \vee b', a' \prox a, b' \prox b$}), \\
    \top  \entails 1,
    \quad\qquad
    a \wedge b \entails (a \wedge b),
    \quad\qquad
    a  \entails b \qquad (\text{if $a \leq b$}),\\
    a  \entails b \qquad (\text{if $a \prox b$}),
    \quad\qquad
    a  \entails \bigvee_{b \prox a}b.
  \end{gathered}
\end{equation}
In the case of continuous entailment relations (or entailment
systems), the presentations of the locales represented by these
structures are more elaborate; see Vickers
\cite[Corollary~44]{VickersEntailmentSystem}.
To obtain a simpler presentation of $\Spec{S}$, we strengthen the
notion of proximity lattice to strong proximity
lattice~\cite{JungSunderhaufDualtyCompactOpen}, which can be obtained
from the well-known fact that stably compact locales are the retracts
of the spectral locales. 

\subsection{Strong proximity lattices} \label{subsec:StrongProximityLattice}
We start from a finitary description of locale maps between  spectral
locales.
\begin{definition}
  Let $S$ and $S'$ be distributive lattices. 
  A proximity relation $r \colon S \to S'$ is said to be
  \emph{join-preserving} if
  \begin{enumerate}
    \myitem[(Prox$0$)]\label{def:proximity0}
    $a \mathrel{r} 0' \imp a = 0$,
    \myitem[(Prox$\vee$)]\label{def:proximityJ}
    $a \mathrel{r} (b \vee' c) \imp \exists b', c' \in S  \left( a \leq b' \vee c' \amp b'
    \mathrel{r} b \amp c' \mathrel{r} c \right)$.
  \end{enumerate}
\end{definition}
Proposition \ref{prop:ProximityRelPreframeHom} restricts 
to join-preserving proximity relations and frame homomorphisms.
\begin{proposition}
  \label{prop:jProximityRelFrameHom}
  For any join-preserving proximity relation $r \colon S \to S'$
  between distributive lattices,
  there exists a unique frame homomorphism $f \colon \Ideals{S'} \to
  \Ideals{S}$ such that $f(\downset b) = r^{-}b$ for all $b \in S'$.
\end{proposition}
\begin{proof}
  The proof is similar to Proposition
  \ref{prop:ProximityRelPreframeHom}.
  One only has to note that a join-preserving proximity relation
  $r \colon S \to S'$ uniquely extends to a lattice homomorphism 
  $f_{r} \colon S' \to \Ideals{S}$ defined as \eqref{eq:ProximityRelMeetHom}.
  The desired conclusion then follows from  Vickers \cite[Theorem
  9.1.5 (iii) (iv)]{vickers1989topology}.
\end{proof}
Let $\DLatjProx$ be the subcategory of $\DLatProx$ where
morphisms are join-preserving proximity relations.
Let $\Spectral$ be the category of spectral locales and locale maps.
The following is immediate from Proposition \ref{prop:jProximityRelFrameHom}.
\begin{theorem}
  \label{thm:DLatjProxEquivSpectral}
  The category $\DLatjProx$ is equivalent to $\Spectral$.
  \qed
\end{theorem}

Since $\Karoubi{\Spectral}$ is equivalent to the category of stably
compact locales and locale maps, $\Karoubi{\DLatjProx}$ is equivalent
to the latter category.  The objects and morphisms of
$\Karoubi{\DLatjProx}$ are called \emph{$\vee$-strong proximity
lattices} and \emph{joint-preserving proximity relations}
respectively (cf.\ van Gool \cite[Definition 1.2 and Definition
1.9]{vanGoolDualityCanExt}).
We write $\jSPxLat$ for $\Karoubi{\DLatjProx}$ and adopt the similar
notation as in Notation \ref{not:ProxLat} for $\vee$-strong proximity
lattices and join-preserving proximity relations between them.

As in the case of proximity lattices, each $\vee$-strong proximity
lattice $(S, \prox)$ represents a stably compact locale 
by the collection $\RIdeals{S}$ of rounded ideals. 
In this case, the locale $\Spec{S}$ can be presented by a
simpler geometric theory than \eqref{eq:GeoTheoryProxLat},
as we now show.
 
Let $X$ and $Y$ be locales. A function $f \colon \Frame{X} \to
\Frame{Y}$ is \emph{Scott continuous}  if it preserves directed joins.
A Scott continuous function is a \emph{suplattice 
homomorphism} if it preserves finite joins (and hence all joins).
\begin{definition}
  Let $(S,\prox)$ be a $\vee$-strong proximity lattice and $X$ be a
  locale.  A \emph{dcpo interpretation} of $(S,\prox)$
  in $X$ is an order preserving function $f \colon S \to
  \Frame{X}$ such that
  $
  f(a) = \bigvee_{b \prox a} f(b).
  $
  A dcpo interpretation $f \colon S \to \Frame{X}$ is called a
  \emph{suplattice (preframe) interpretation} if it preserves finite
  joins (resp.\ finite meets); $f$ is called an \emph{interpretation} if it
  preserves both finite joins and finite meets.
\end{definition}

Any $\vee$-strong proximity lattice $(S,\prox)$ admits
an interpretation $i_{S} \colon S \to \RIdeals{S}$ defined by
\begin{equation}\label{eq:InjectionGenerator}
  i_{S}(a) \defeql \downset_{\prox} a.
\end{equation}

\begin{proposition}
  \label{prop:BijInterpretationScottCont}
  Let $(S,\prox)$ be a $\vee$-strong proximity lattice and $X$ be a locale.
  For any (dcpo, suplattice, preframe) interpretation $f \colon S \to
  \Frame{X}$, there exists a unique frame homomorphism (resp.\ Scott
  continuous function, suplattice homomorphism, preframe homomorphism)
  $\overline{f} \colon \RIdeals{S} \to \Frame{X}$ such that
  $\overline{f} \circ i_{S} = f$.
\end{proposition}
\begin{proof}
  See Vickers \cite[Theorem 9.1.5]{vickers1989topology} where
  an analogous fact for spectral locales is presented.
  The unique extension of $f$ is defined by
  $
  \overline{f}(I) \defeql \bigvee_{a \in I} f(a)
  $. Since $f$ is a dcpo interpretation,
  we have $\overline{f} \circ i_{S} = f$.
\end{proof}

\begin{remark}
  \label{rem:DcpoPreframeInterpretation}
  For dcpo and preframe interpretations, Proposition
  \ref{prop:BijInterpretationScottCont} holds for proximity lattices
  as well. In this case, however, the function 
  $i_{S} \colon S \to \RIdeals{S}$ does not necessarily preserve
  finite joins, so it is only a preframe interpretation.
\end{remark}

\begin{corollary}
  \label{cor:SProxLatPresentation}
  For any $\vee$-strong proximity lattice $(S, \prox)$, the locale
  $\Spec{S}$ is presented by a geometric theory over $S$ with the
  following axioms:\footnote{The left part of axiom $(a \vee b) \entails
  a \vee b$ denotes the generator $a \vee b$, while the right part is
the disjunction of two generators $a$ and $b$. The similar remark
applies to $a \wedge b \entails (a \wedge b)$.}
  \begin{gather*}
    0 \entails \bot,
    \qquad
    (a \vee b) \entails a \vee b,
    \qquad
    \top  \entails 1,
    \qquad
    a \wedge b \entails (a \wedge b),
    \qquad
    a \entails b \;\; (\text{if $a \leq b$}),\\
    \notag
    a  \entails b \;\; (\text{if $a \prox b$}),
    \qquad \quad
    a  \entails \bigvee_{b \prox a}b.
  \end{gather*}
\end{corollary}
\begin{proof}
  Immediate from the frame version of Proposition
\ref{prop:BijInterpretationScottCont}.
\end{proof}
Note that models of the geometric theory in Corollary
\ref{cor:SProxLatPresentation}
are rounded prime filters of $(S,\prox)$, i.e., those prime filters
$F$ on $S$ such that $a \in F \leftrightarrow \exists b \prox a \left( b \in F
\right)$. 

Let $\jSPxLatP$ be the full subcategory of $\ProxLat$ consisting of
$\vee$-strong proximity lattices. 
\begin{theorem}
  \label{prop:EquivPxLatandjSPxLat}
  The category $\jSPxLatP$ is equivalent to $\ProxLat$.
\end{theorem}
\begin{proof}
  Given a  proximity lattice $(S, \prox)$, 
  define a preorder $\leq^{\vee}$ on $\Fin{S}$ by
  \begin{align}
    \notag
    A \leq^{\vee} B
    &\defeqiv
    \forall C \prox_{L} A \exists D \prox_{L} B
    \left( \medvee C \prox \medvee D \right),\\
    \shortintertext{where}
    \label{def:LowerExtension}
    A \prox_{L} B 
    &\defeqiv
    \forall a \in A \exists b \in B \left( a \prox b\right).
  \end{align}
  Let $S^{\vee}$ be the set $\Fin{S}$ equipped with the equality
  $=^{\vee}$ determined by $\leq^{\vee}$, i.e., ${=^{\vee}} \defeql
  {\leq^{\vee} \cap \geq^{\vee}}$, 
  and define a lattice structure
  $(S^{\vee}, 0^{\vee},\vee^{\vee}, 1^{\vee}, \wedge^{\vee})$ by 
  \begin{equation}
    \label{eq:LatticeStructVee}
    \begin{aligned}
      0^{\vee} &\defeql \emptyset,
      \qquad&
      A \vee^{\vee} B &\defeql A \cup B,\\
      1^{\vee} &\defeql \left\{ 1 \right\}, 
      \qquad&
      A \wedge^{\vee} B &\defeql \left\{ a \wedge b \mid a \in A, b \in B \right\}.
    \end{aligned}
  \end{equation}
  It is straightforward to check that the above operations respect
  $=^{\vee}$ and that the lattice is distributive.
  Next, define a relation $\prox^{\vee}$ on $\Fin{S}$ by
  \[
    A \prox^{\vee} B 
    \defeqiv
    \exists C \prox_{L} B 
    \left( A \leq^{\vee} C \right).
  \]
  Again, it is straightforward to check that $\prox^{\vee}$ respects $=^{\vee}$.
  Since $\prox_{L}$ is idempotent and $\prox$ is a proximity relation,
  $\prox^{\vee}$ is an idempotent relation.
  We claim that 
  $(S^{\vee}, \prox^{\vee}) 
  \defeql
  (S^{\vee}, 0^{\vee},\vee^{\vee}, 1^{\vee}, \wedge^{\vee}, \prox^{\vee})$ 
  is a $\vee$-strong proximity lattice. For example, to see that
  $(S^{\vee}, \prox^{\vee})$ is $\vee$-strong, suppose that
  $A \prox^{\vee} B \cup C$. Then, there exist $B' \prox_{L} B$ and
  $C' \prox_{L} C$ such that
  $A \leq^{\vee} B' \vee^{\vee} C'$. Since ${\prox_{L}} \subseteq
  {\prox^{\vee}}$, we have $B' \prox^{\vee} B$ and $C \prox^{\vee} C$.
  Moreover, $A \prox^{\vee} 0^{\vee}$ clearly implies
  $A \leq^{\vee} 0^{\vee}$.

  Define  relations $r \subseteq S \times \Fin{S}$ and 
  $s \subseteq \Fin{S} \times S$ by
  \begin{align*}
    a \mathrel{r} A 
    &\defeqiv \exists B \prox_{L} A 
    \left( a \prox \medvee B\right), &
    A \mathrel{s} a \
    &\defeqiv A \prox^{\vee} \left\{ a \right\}.
  \end{align*}
  Then, $r$ and $s$ clearly respect $=^{\vee}$,
  so they are relations between $S$ and $S^{\vee}$.
  It is straightforward to show that $r$ and $s$ are proximity
  relations between $(S, \prox)$ and $(S^{\vee}, \prox^{\vee})$ and
  are inverse to each other.
\end{proof}

By introducing $\vee$-strong proximity lattices, we have obtained a
simpler geometric theory for $\Spec{S}$ (cf.\
Corollary~\ref{cor:SProxLatPresentation}). This, however, comes at the
cost of the structural duality of proximity lattices. Nevertheless, the
notion of $\vee$-strong proximity lattice is categorically equivalent
to the one with a stronger self-dual structure than that of a proximity
lattice.

\begin{definition}[{Jung and S\"underhauf~\cite[Definition 18]{JungSunderhaufDualtyCompactOpen}}]
A \emph{strong proximity lattice} is 
a $\vee$-strong proximity lattice $(S, \prox)$ satisfying
\begin{enumerate}
  \myitem[(Prox1)]\label{def:proximity1}  $1 \prox a \imp a = 1$,
  \myitem[(Prox$\wedge$)]\label{def:proximityMeet}  
     $b \wedge c \prox a \imp
    \exists b' \proxop b \exists c'  \proxop c 
    \left(b' \wedge c' \leq a\right)$.
\end{enumerate}
\end{definition}

Let $\SProxLat$ and $\SProxLatP$ be the full subcategories of $\jSPxLat$
and $\jSPxLatP$, respectively, consisting of strong proximity lattices.
In Section~\ref{subsec:GenetateContEnt}, we show that $\SProxLat$
and $\SProxLatP$ are equivalent to the larger categories.
\subsection{Strong continuous entailment relations}\label{sec:StrongContEnt}
We characterise a full subcategory of $\ContEnt$, whose objects
correspond to strong proximity lattices.
The notion introduced below is a 
modification of that of \emph{entailment relation with the interpolation
property} by Coquand and Zhang \cite{CoquandZhangPrediativePatch},
which satisfies only one direction of \eqref{eq:ContEnt}.
\begin{definition}\label{def:ContEnt}
  A \emph{strong continuous entailment relation}
  is an entailment relation $(S, \entails)$ equipped
  with an idempotent relation ${\prox} 
  \subseteq S \times S$ satisfying
  \begin{equation}\label{eq:ContEnt}
     \exists A' \in \Fin{S} \left( A \prox_{U} A' \entails B\right)
    \iff
     \exists B' \in \Fin{S} \left(  A \entails B' \prox_{L}  B\right)
  \end{equation}
  for all $A,B \in \Fin{S}$ where $\prox_{L}$ is defined as
  \eqref{def:LowerExtension} and $\prox_{U}$ is defined by
  \[
    A \prox_{U} B \defeqiv  \forall b \in B \exists a \in A \left( a \prox b\right).
  \]
  We write $(S, \entails, \prox)$ for a strong continuous entailment relation.
\end{definition}

Each strong continuous entailment relation $(S,\entails,\prox)$ represents a
stably compact locale by a geometric theory
$T(S,\entails,\prox)$ over $S$ with the following axioms:%
\begin{equation}
  \label{eq:SContEntGeo}
  \medwedge A \entails \medvee B \quad (\text{if $A \entails B$}),
  \qquad
  a \entails b \quad (\text{if $a \prec b$}),
  \qquad
  a \entails \bigvee_{b \prox a}  b.
\end{equation}
\begin{theorem}[Coquand and Zhang \cite{CoquandZhangPrediativePatch}]\label{thm:ContEntPresentsSKL}
  For any strong continuous entailment relation $(S,\entails, \prox)$,
  the locale presented by $T(S,\entails,\prox)$ is stably compact.
  Moreover, any stably compact locale can be presented in this way.
\end{theorem}
\begin{proof}
  See Coquand and Zhang \cite[Theorem 1]{CoquandZhangPrediativePatch}.
\end{proof}
In the following, we often identify a strong continuous entailment relation
${(S,\entails,\prox)}$ with the theory $T(S,\entails,\prox)$.

We relate strong continuous entailment relations to continuous
entailment relations.
\begin{lemma}
  \label{lem:AssocCutComposition}
  Let $S,S',S'', S'''$ be sets and $r,s, t$ be relations 
  $r \subseteq \Fin{S} \times \Fin{S'}$,
  $s \subseteq \Fin{S'} \times \Fin{S''}$, and 
  $t \subseteq \Fin{S''} \times \Fin{S'''}$.
  \begin{enumerate}
    \item \label{lem:AssocCutComposition1}
      If $s$ is upper and $rA = \left\{ B \in \Fin{S'} \mid A
      \mathrel{r} B \right\}$ is closed under finite joins for each
      $A \in \Fin{S}$, then
    $
    \left( t \cdot s  \right) \circ r 
    =  t \cdot \left( s \circ r \right).
    $

    \item \label{lem:AssocCutComposition2}
      If $s$ is upper and $t^{-}D = \left\{ C \in \Fin{S''} \mid C
      \mathrel{t} D \right\}$ is closed under finite joins for each
      $D \in \Fin{S'''}$, then
    $
    \left( t \circ s  \right) \cdot r
    =  t \circ \left( s \cdot r \right).
    $
  \end{enumerate}
\end{lemma}
\begin{proof}
  \noindent\ref{lem:AssocCutComposition1}.
  Suppose that
  $A \mathrel{t \cdot \left( s \circ r \right)} D$.
  Then there exists $\mathcal{V}
  \in \Fin{\Fin{S''}}$ such that $\forall C' \in \mathcal{V}^{*}
  \left( A \mathrel{s \circ r} C' \right)$ and $\forall C \in
  \mathcal{V} \left( C \mathrel{t} D\right)$,
  so for each $C' \in \mathcal{V}^{*}$ there exists $B_{C'} \in
  \Fin{S'}$ such that $A \mathrel{r} B_{C'}$ and $B_{C'} \mathrel{s} C'$.
  Put $B = \bigcup_{C' \in \mathcal{V}^{*}}B_{C'}$. Then
  $A \mathrel{r} B$ and $\forall C' \in
  \mathcal{V}^{*} \left( B \mathrel{s} C' \right) $, and hence $A
  \mathrel{\left( t \cdot s  \right) \circ r} D$. The converse is easy. 

  \noindent\ref{lem:AssocCutComposition2}. The proof is dual of \ref{lem:AssocCutComposition1}.
\end{proof}
For each strong continuous entailment relation $(S, \entails, \prox)$,
define a relation $\ll_{\entails}$ on $\Fin{S}$ by
\[
  \ll_{\entails}  {\defeql} \entails \circ \prox_{U} {=} \prox_{L} \circ
  \entails.
\]
\begin{proposition}
  \label{prop:SContEntIsContEnt}
  The structure $(S, \entails, \ll_{\entails})$ is a continuous entailment relation.
\end{proposition}
\begin{proof}
   By item \ref{lem:AssocCutComposition1} of Lemma~\ref{lem:AssocCutComposition},
   we have
   \[
     {\entails} \cdot {\ll_{\entails}}
     = {\entails} \cdot \left( \entails \circ \prox_{U} \right) 
     = \left( {\entails}  \cdot {\entails} \right) \circ {\prox_{U}}
     = {\entails} \circ {\prox_{U}}
     = {\ll_{\entails}}.
   \]
    Similarly ${\ll_{\entails}} \cdot {\entails} = {\ll_{\entails}}$
    by item \ref{lem:AssocCutComposition2} of
    Lemma~\ref{lem:AssocCutComposition}. Then
    \[
      {\ll_{\entails}} \cdot {\ll_{\entails}}
      =  \left( {\prox_{L}} \circ {\entails} \right) \cdot {\ll_{\entails}}
      =  {\prox_{L}} \circ  \left( {\entails} \cdot {\ll_{\entails}} \right)
      =  {\prox_{L}} \circ   {\ll_{\entails}} 
      = {\ll_{\entails}}.
    \]
   Hence ${\ll_{\entails}}$ is an idempotent Karoubi endomorphism on
   $(S, \entails)$.
\end{proof}
Let $\SContEntP$ be a category where objects are strong continuous
entailment relations and morphisms are proximity maps between
the underlying continuous entailment relations. By the assignment
$(S,\entails, \prox) \mapsto (S, \entails, \ll_{\entails})$, we can
identify $\SContEntP$ with a full subcategory of $\ContEnt$.

In what follows, we show that $\SContEntP$ 
and $\SProxLatP$ are equivalent.
First, note that 
the functor $E
\colon \EntRel \to \DLatProx$ (cf.\ Proposition
\ref{prop:EntRelEquivDLatProx}) induces a functor
\[
  F \colon \ContEnt \to \ProxLat,
\]
which establishes equivalence of $\ContEnt$ and $\ProxLat$.
The functor $F$ sends each continuous entailment relation $(S,
\entails, \ll)$ to a proximity lattice $(\EnttoLat{S, \entails},
\ApproxExt{\ll})$ and each proximity map $r$ to
a proximity relation $\ApproxExt{r}$. By Remark \ref{rem:QuasiInverseofE},
 $F$ has a quasi-inverse 
$G \colon \ProxLat \to \ContEnt$ which sends each proximity
lattice $(S, \prox)$ to a continuous entailment relation
$(S, \entails, \ll)$, where $\entails$ is given by \eqref{eq:EntofDLat} 
and $\ll$ is defined by
\begin{equation}
  \label{eq:ProxLatToContEntll}
  A \ll B \defeqiv \medwedge A \prox \medvee B.
\end{equation}
\begin{lemma}
  \label{lem:SContEntToSProxLat}
    For each strong continuous entailment relation $(S, \entails ,
    \prox)$, the structure
    $F(S, \entails, \ll_{\entails}) = (\EnttoLat{S,\entails},
    \ApproxExt{\ll_{\entails}})$ is a strong proximity lattice.
\end{lemma}
\begin{proof}
   We must show that 
   $(\EnttoLat{S,\entails}, \ApproxExt{\ll_{\entails}})$
   satisfies \ref{def:proximity0}, \ref{def:proximityJ},
   \ref{def:proximity1}, and \ref{def:proximityMeet}.
   As a demonstration, we show \ref{def:proximityJ}.
   Suppose that $\mathcal{U} \ApproxExt{\ll_{\entails}}
    \mathcal{V} \vee \mathcal{W}$.
    For each $B \in \mathcal{V}^{*}$ and $C \in \mathcal{W}^{*}$, there
    exist $B' \prox_{L} B$ and $C' \prox_{L} C$ such that
      $A \entails B' \cup C'$ for all $A \in \mathcal{U}$.
    Put
    $\mathcal{V}' = \left\{ B' \mid B \in \mathcal{V}^{*} \right\}^{*}$ and 
    $\mathcal{W}' = \left\{ C' \mid C \in \mathcal{W}^{*} \right\}^{*}$.
    Then, we have $\mathcal{U} \ApproxExt{\entails} \mathcal{V}' \vee
    \mathcal{W}'$,  $\mathcal{V}' \ApproxExt{\ll_{\entails}} \mathcal{V}$ 
    and $\mathcal{W}' \ApproxExt{\ll_{\entails}} \mathcal{W}$.
\end{proof}
\begin{lemma}
  \label{lem:SProxLatSContEnt}
  For each strong proximity lattice $(S, \prox)$, the structure
  $(S, \entails, \prox)$, where $\entails$ is defined by
  \eqref{eq:EntofDLat}, is a strong continuous entailment
  relation. Moreover the relation $\ll_{\entails}$ determined by
  $(S, \entails, \prox)$ is characterised by \eqref{eq:ProxLatToContEntll}.
\end{lemma}
\begin{proof}
  Straightforward.
\end{proof}
  By Lemma \ref{lem:SContEntToSProxLat} and Lemma
  \ref{lem:SProxLatSContEnt},
  we have the following.
\begin{theorem}
  \label{thm:EuqivSContEntPAndSProxLatP}
  The functor  $F \colon \ContEnt \to \ProxLat$ restricts to 
  $\SContEntP$ and $\SProxLatP$, which establishes
  equivalence of the latter two categories.
  \qed
\end{theorem}

The following corresponds to the notion of join-preserving proximity
relation.
\begin{definition}
  \label{def:JPProxMap}
  Let $(S, \entails, \prox)$ and $(S', \entails', \prox')$ be strong
  continuous entailment relations. A proximity map $r \colon (S,
  \entails, \prox) \to (S', \entails', \prox')$ is 
  \emph{join-preserving} if
  \begin{enumerate}
      \myitem[(JP)]\label{def:proximityStrongJoin} 
      $A \mathrel{r} B \imp
      \exists \, \mathcal{U} \in \Fin{\Fin{S}} 
      \left( 
      \left\{ A \right\}\ApproxExt{\entails} \mathcal{U}\amp
      \forall A' \in \mathcal{U}  \exists b \in
      B  \left( A' \mathrel{r}  \left\{ b \right\} \right) \right)$.
  \end{enumerate}
\end{definition}
Since join-preserving proximity maps are closed under composition
(see the remark below Proposition \ref{prop:PxJFullAndFaithful}),
strong continuous entailment relations and join-preserving proximity
maps form a subcategory $\SContEnt$ of $\SContEntP$.

Let $\One = (\emptyset, \meets, =)$ be a terminal object in $\SContEnt$.
The notion of join-preserving proximity map is consistent with the theory
$T(S, \entails, \prox)$ in~\eqref{eq:SContEntGeo}. 
\begin{proposition}
  \label{prop:SContEntPoint}
  For any strong continuous entailment relation $(S, \entails,
  \prox)$, there exists a bijective correspondence between 
  the models of $T(S, \entails, \prox)$ and the join-preserving 
  proximity maps from $\One$ to $S$.
\end{proposition}
\begin{proof}
  A model $\alpha$ of $T(S, \entails, \prox)$ corresponds
  to a join-preserving proximity map $r_{\alpha} \colon \One \to S$
  defined by
  \[
    \emptyset \mathrel{r_{\alpha}} A \defeqiv \alpha \meets A.
  \]
  Conversely, a join-preserving proximity map $r \colon \One \to S$ 
  corresponds to a model $\alpha_{r}$ of ${T(S, \entails, \prox)}$
  defined by
  \[
    \alpha_{r} \defeql \left\{ a \in S \mid \emptyset
      \mathrel{r} \left\{ a \right\} \right\}.
  \]
  It is straightforward to check that the above correspondence is bijective.
\end{proof}

We now restrict Theorem \ref{thm:EuqivSContEntPAndSProxLatP} to
$\SContEnt$ and $\SProxLat$. The following should be compared with
Vickers~\cite[Theorem 42]{VickersEntailmentSystem}.
\begin{lemma}
  \label{lem:CharPxJ}
  The condition \textup{\ref{def:proximityStrongJoin}} is equivalent to 
  the following:
  \begin{enumerate}
      \myitem[{\textup{(JP$0$)}}]\label{def:proximityStrongJoin0} 
      $A \mathrel{r} \emptyset \imp A \entails \emptyset$, 

      \myitem[\textup{(JP$\vee$)}]\label{def:proximityStrongJoin1} 
      $A \mathrel{r} B \cup C \imp 
      \exists \, \mathcal{U}, \mathcal{V} \in \Fin{\Fin{S}} 
      \left( 
      \left\{ A \right\} \ApproxExt{\entails} \mathcal{U} \cup
      \mathcal{V} 
      \amp
      \mathcal{U} \ApproxExt{r} \left\{ B \right\}^{*}
      \amp
      \mathcal{V} \ApproxExt{r} \left\{ C \right\}^{*} \right)$.
  \end{enumerate}
\end{lemma}
\begin{proof}
  Assume \ref{def:proximityStrongJoin}.  For
  \ref{def:proximityStrongJoin0}, if $A \mathrel{r} \emptyset$, then
  we must have $\left\{ A \right\} \ApproxExt{\entails} \emptyset$ 
  so $A \entails
  \emptyset$. For \ref{def:proximityStrongJoin1}, suppose that 
  $A \mathrel{r} B \cup C$. By 
  \ref{def:proximityStrongJoin}, there exist $\mathcal{U},
  \mathcal{V} \in
  \Fin{\Fin{S}}$ such that $\left\{ A \right\} \ApproxExt{\entails}
  \mathcal{U} \cup \mathcal{V}$, and 
  $\forall B' \in \mathcal{U}
  \exists b \in B  \left( B' \mathrel{r} \left\{ b \right\}  \right)$ 
  and 
  $\forall C' \in \mathcal{V}
  \exists c \in C  \left( C' \mathrel{r}  \left\{ c \right\} \right)$.
  Then,
  $\forall B' \in \mathcal{U} \left( B' \mathrel{r} B\right)$ 
  and
  $\forall C' \in \mathcal{V} \left( C' \mathrel{r} C\right)$.

  Conversely, assume
  \ref{def:proximityStrongJoin0} and \ref{def:proximityStrongJoin1}.
  We show \ref{def:proximityStrongJoin} by induction on the size of
  $B$. The base case $B = \emptyset$ follows from
  \ref{def:proximityStrongJoin0}. For the inductive case, suppose that 
  $A \mathrel{r} B \cup \left\{ b \right\}$.
  By \ref{def:proximityStrongJoin1}, there exist $\mathcal{U},
  \mathcal{V} \in \Fin{\Fin{S}}$ such that $A \ApproxExt{\entails}
  \mathcal{U} \cup \mathcal{V}$, 
  $\mathcal{U} \ApproxExt{r} \left\{B \right\}^{*}$,
  and $\mathcal{V} \ApproxExt{r} \left\{ \left\{ b \right\}
\right\}^{*}$.
  By induction hypothesis, for each $C \in \mathcal{U}$ there exists 
  $\mathcal{U}_{C} \in \Fin{\Fin{S}}$ such that 
  $C \ApproxExt{\entails} \mathcal{U}_{C}$
  and  $\forall B' \in \mathcal{U}_{C} \exists b' \in B \left( B'
  \mathrel{r} \left\{ b' \right\} \right)$. Then $\bigcup_{C \in
    \mathcal{U}}\mathcal{U}_{C} \cup \mathcal{V}$ witnesses
    \ref{def:proximityStrongJoin} for $B \cup \left\{ b
    \right\}$.
\end{proof}

\begin{proposition}
  \label{prop:PxJFullAndFaithful}
  A proximity map $r \colon (S, \entails, \prox) \to  (S', \entails', \prox') $
  is join-preserving if and only if 
  $\ApproxExt{r} \colon (\EnttoLat{S,\entails}, \ApproxExt{\ll_{\entails}})
  \to (\EnttoLat{S',\entails'}, \ApproxExt{\ll_{\entails'}})$ is
  join-preserving.
\end{proposition}
\begin{proof}
  Suppose that $r$ is join-preserving.

  \noindent \ref{def:proximity0}
  Suppose $\mathcal{U} \mathrel{\ApproxExt{r}} \emptyset$.
  By \ref{def:proximityStrongJoin0}, we have $A \entails \emptyset$ for all
  $A \in \mathcal{U}$. Thus  $\mathcal{U} \ApproxExt{\entails} \emptyset$.

  \noindent \ref{def:proximityJ}
  Suppose $\mathcal{U} \ApproxExt{r} \mathcal{V} \vee
  \mathcal{W}$. Since
  $(\mathcal{V} \vee \mathcal{W})^{*} =_{\entails'} \mathcal{V}^{*} \wedge
  \mathcal{W}^{*}$ in $\EnttoLat{S',\entails'}$,
  for each  $A \in \mathcal{U}$,
  $B \in \mathcal{V}^{*}$, and $C \in \mathcal{W}^{*}$,
  we have $A \mathrel{r} B \cup C$. By
  \ref{def:proximityStrongJoin1}, there exist 
  $\mathcal{V}_{A,B,C}, \mathcal{W}_{A,B,C}  \in \Fin{\Fin{S}}$
  such that $A \ApproxExt{\entails} \mathcal{V}_{A,B,C} \cup
  \mathcal{W}_{A,B,C}$,
  $\mathcal{V}_{A,B,C} \ApproxExt{r} \left\{ B \right\}^{*}$,
  and 
  $\mathcal{W}_{A,B,C} \ApproxExt{r} \left\{ C \right\}^{*}$.
  Put
  \begin{align*}
    \mathcal{V}' &=
    \bigvee_{A \in \mathcal{U}}
    \bigwedge_{B \in \mathcal{V}^{*}}
    \bigvee_{C \in \mathcal{W}^{*}}
    \mathcal{V}_{A,B,C}, &
    \mathcal{W}' &=
    \bigvee_{A \in \mathcal{U}}
    \bigwedge_{C \in \mathcal{W}^{*}}
    \bigvee_{B \in \mathcal{V}^{*}}
    \mathcal{W}_{A,B,C}.
  \end{align*}
  Then,
  $\mathcal{V}' \mathrel{\ApproxExt{r}} \mathcal{V}$
  and 
  $\mathcal{W}' \mathrel{\ApproxExt{r}} \mathcal{W}$.
  Since %
  \[
    \left\{ A \right\} 
    \ApproxExt{\entails}
    \bigwedge_{B \in \mathcal{V}^{*}}
    \bigwedge_{C \in \mathcal{W}^{*}}\!\!
    \mathcal{V}_{A,B,C} \cup \mathcal{W}_{A,B,C} %
    \ApproxExt{\entails}
    \bigwedge_{B \in \mathcal{V}^{*}}
    \bigvee_{C \in \mathcal{W}^{*}}
    \mathcal{V}_{A,B,C}
    \vee
    \bigwedge_{C \in \mathcal{W}^{*}}
    \bigvee_{B \in \mathcal{V}^{*}}
    \mathcal{W}_{A,B,C}
  \]
  for each $A \in \mathcal{U}$, we have $\mathcal{U} \ApproxExt{\entails}
  \mathcal{V}' \cup \mathcal{W}'$.

  Conversely, suppose that $\ApproxExt{r}$ is join-preserving.
  We show
  \ref{def:proximityStrongJoin0}
  and \ref{def:proximityStrongJoin1}.

  \noindent \ref{def:proximityStrongJoin0}
  Suppose $A \mathrel{r} \emptyset$. Then $\left\{ A
  \right\} \mathrel{\ApproxExt{r}} \emptyset$. Thus $\left\{ A
  \right\} \ApproxExt{\entails} \emptyset$ by \ref{def:proximity0},
  and so $A \entails \emptyset$.

  \noindent \ref{def:proximityStrongJoin1}
  Suppose $A \mathrel{r} B \cup C$.
  Then $\left\{ A \right\} \mathrel{\ApproxExt{r}} \left\{ B
  \right\}^{*} \vee \left\{ C \right\}^{*}$.
  By \ref{def:proximityJ}, there exist
  $\mathcal{U}, \mathcal{V} \in \Fin{\Fin{S}}$
  such that $\left\{ A \right\} \ApproxExt{\entails} \mathcal{U} \cup
  \mathcal{V}$, 
  $\mathcal{U} \mathrel{\ApproxExt{r}} \left\{ B \right\}^{*}$, and 
  $\mathcal{V} \mathrel{\ApproxExt{r}} \left\{ C \right\}^{*}$.
\end{proof}

In particular, since join-preserving proximity relations are closed
under composition, so do join-preserving proximity maps.
\begin{theorem}
  \label{thm:EquiSContEntSProxLat}
  The categories $\SContEnt$ and $\SProxLat$ are equivalent.
\end{theorem}
\begin{proof}
  By Theorem \ref{thm:EuqivSContEntPAndSProxLatP}
  and Proposition
  \ref{prop:PxJFullAndFaithful}, 
  the functor $F \colon \ContEnt \to \ProxLat$ restricts to 
  a full and faithful functor from $\SContEnt$ to $\SProxLat$.
  Since every isomorphic
  proximity relation between $\vee$-strong proximity lattices 
  is join-preserving, $F$ establishes an equivalence of 
  $\SContEnt$ and $\SProxLat$.
\end{proof}

By an abuse of notation, we write $F \colon \SContEnt \to \SProxLat$
and $G \colon \SProxLat \to \SContEnt$ for the restrictions of the
functor  $F \colon \ContEnt \to \ProxLat$  and its quasi-inverse $G \colon
\ProxLat \to \ContEnt$.
\begin{remark}
  \label{rem:EquivSContEntSProxLat}
  Many of the examples in Section~\ref{sec:ExampledeGrootDuality}
  start from a strong proximity lattice $(S, \prox)$ and specify a
  strong continuous entailment relation which represents the desired
  construction on $\Spec{S}$. The functor $F \colon \SContEnt \to
  \SProxLat$ then allows us to calculate the corresponding construction
  on $(S,\prox)$.
\end{remark}

The presentations of stably compact locales are invariant
under the equivalence of $\SContEnt$ and $\SProxLat$ in the following sense.
\begin{proposition}\label{prop:AgreeContEntProxLat}
  \leavevmode
  \begin{enumerate}
    \item\label{prop:AgreeContEntProxLat1} For any strong proximity
      lattice $(S,\prox)$, the locale
      $\Spec{S}$ is presented by $G(S,\prox)$.
  \item\label{prop:AgreeContEntProxLat2} For any continuous entailment
    relation $(S,\entails,
    \prox)$, the locale ${\Spec{F(S,\entails,\prox)}}$ is presented by
  $(S,\entails, \prox)$.
  \end{enumerate} 
\end{proposition}
\begin{proof}
 \ref{prop:AgreeContEntProxLat1}. This is clear from the definition of
 $G(S,\prox)$ and Corollary \ref{cor:SProxLatPresentation}.

  \noindent \ref{prop:AgreeContEntProxLat2}.  First, we define a bijection 
  between interpretations of $(S,\entails, \prox)$ in a locale $X$
  and interpretations of $GF(S,\entails,\prox)$ in $X$
  via a mapping $a \mapsto \left\{ \left\{ a \right\}
\right\} \colon S \to \Fin{\Fin{S}}$.
  Let $f \colon S \to \Frame{X}$ be an interpretation of $(S,\entails,
  \prox)$ in $X$. Define $\overline{f} \colon \Fin{\Fin{S}}
  \to \Frame{X}$ by
  \[
    \overline{f}(\mathcal{U})
    \defeql
    \medvee_{A \in \mathcal{U}} \medwedge_{a \in A} f(a),
  \]
  which clearly satisfies
  $\overline{f}(\left\{ \left\{ a \right\} \right\}) = f(a)$ for all
  $a \in S$.
  We show that $\overline{f}$ preserves the order on
  $\EnttoLat{S,\entails}$,
  which implies that $\overline{f}$ respects the equality on
  $\EnttoLat{S,\entails}$.
  Suppose $\mathcal{U} \ApproxExt{\entails} \mathcal{V}$. Since $f$ 
  is an interpretation of $(S,\entails, \prox)$, we have
  \[
    \overline{f}(\mathcal{U}) 
    =
    \medvee_{A \in \mathcal{U}} \medwedge_{a \in A} f(a)
    \leq_{X}
    \medwedge_{B' \in \mathcal{V}^{*}} \medvee_{b' \in B'} f(b')
    =
    \medvee_{B \in \mathcal{V}} \medwedge_{b \in B} f(b)
    = \overline{f}(\mathcal{V}),
  \]
  where $\leq_{X}$ is the order on $X$. Thus, $\overline{f}$ is a
  function on $\EnttoLat{S, \entails}$.
  Similarly, we have
  $
  \mathcal{U} \ApproxExt{\ll_{\entails}} \mathcal{V} \imp \overline{f}(\mathcal{U}) \leq_{X} \overline{f}(\mathcal{V})
  $.
  It is also easy to check
  that $\overline{f}$ preserves finite meets and finite joins. 
  Furthermore, for any $A \in \Fin{S}$,  we have
  \[
    \medwedge_{a \in A} f(a) 
    = \medwedge_{a \in A} \bigvee_{b \prox a} f(b)
    = \bigvee_{B \prox_{U} A} \medwedge_{b \in B} f(b),
  \]
  which implies 
  $\overline{f}(\mathcal{U}) \leq_{X}\! \bigvee_{\mathcal{V}
  \ApproxExt{\ll_{\entails}}
  \mathcal{U}} \overline{f}(\mathcal{V})$.
  Thus, $\overline{f}$ is an interpretation of
  $GF(S,\entails,\prox)$ in $X$.
  Since $\mathcal{U} =_{\entails} \medvee_{A \in \mathcal{U}} \medwedge_{a \in
  A}\left\{ \left\{ a \right\} \right\}$ for each $\mathcal{U} \in
  \Fin{\Fin{S}}$, $\overline{f}$ is a unique
  interpretation of  $GF(S,\entails,\prox)$ in
  $X$ such that $\overline{f}(\left\{ \left\{ a \right\} \right\}) =
  f(a)$ for all $a \in S$.

  Define $j_{S} \colon S \to \RIdeals{F(S,\entails,\prox)}$ by
  $j_{S}(a) = \downset_{\ApproxExt{\ll_{\entails}}}\left\{ \left\{ a \right\} \right\}$.
  Then, it is straightforward to show that $j_S$
   is a universal interpretation of $(S,\entails,\prox)$.
\end{proof}

\subsection{Generated strong continuous entailment
relations}\label{subsec:GenetateContEnt}
To construct a new entailment relation, one often specifies
a set of initial entailments from which the entire relation
is generated.
\begin{definition}
  An \emph{axiom} on a set $S$ is a pair $(A,B) \in \Fin{S} \times \Fin{S}$.
  Given a set $\entails_{0}$ of axioms on $S$,  an entailment relation $(S,\entails)$ 
  is said to be \emph{generated} by $\entails_0$ if $\entails$ is the
  smallest entailment relation on $S$ that contains $\entails_0$.

  We usually write $A \entails_0 B$ for $(A,B) \mathop{\in} \entails_0$.
\end{definition}

\begin{lemma} \label{lem:IndGenEnt}
  If $\entails_0$ is a set of axioms on a set $S$, then the entailment relation
  $\entails$ generated by $\entails_0$ is inductively defined by the following
  rules:
\[
  \frac{A \meets B}{A \entails B} \;(\mathrm{R'})
  \qquad
  \qquad
  \frac{A \entails_0 C \quad   \forall c \in C \left(  A',c \entails B\right) }
  {A,A' \entails B} \;(\mathrm{AxL})
\]
\end{lemma}
\begin{proof}
  First, we show that the relation $\entails$ generated by 
  $(\mathrm{R'})$ and $(\mathrm{AxL})$ is an entailment relation.
  The proof is by induction on the height of
  derivations of the premises of each condition in Definition
  \ref{def:EntRel}. For example, to see
  that $\entails$ satisfies ($\mathrm{T}$), we show that
  $\entails$ satisfies more general condition:
\[
  \frac{A \entails B,a \quad a,A' \entails B'}{A,A' \entails B,B'} \;
  (\mathrm{T'})
\]
  Suppose $A \entails
  B, a$ and $a, A' \entails B'$.  Then $A \entails B,a$ is derived by
  either ($\mathrm{R'}$) or ($\mathrm{AxL}$).
  The former case is easy. In the
  latter case, $A \entails B, a$ is of the form $C',C \entails B, a$
  for some $C \entails_{0} D$ such that $\forall d \in D \left(  C',
  d \entails B, a \right)$.
  By induction hypothesis, we have $A', C',d \entails B,B'$ for
  all $d \in D$. Hence $A', C, C' \entails B,B'$ by ($\mathrm{AxL}$).
  Next, if $\entails'$ is another entailment relation on $S$ containing
  $\entails_{0}$, then $\entails'$ satisfies
  ($\mathrm{R'}$) and ($\mathrm{AxL}$), so $\entails'$ must contain $\entails$.
\end{proof}
Dually, we have the following.
\begin{lemma}\label{lem:IndGenEntDual}
  If $\entails_0$ is a set of axioms on a set $S$,  the entailment relation
  $\entails$ generated by $\entails_0$ is inductively defined by the following
  rules:
\[
  \frac{A \meets B}{A \entails B} \;(\mathrm{R'})
  \qquad
  \qquad
  \frac{ \forall c \in C \left(  A \entails B',c\right) \quad C \entails_0 B }
  {A \entails B', B} \;(\mathrm{AxR})
\]
\par \vspace{-1.6\baselineskip}
\qed
\end{lemma}

The following is useful when defining a new strong continuous entailment
relation using axioms.
\begin{lemma}\label{prop:IndGenContEnt}
  Let $\entails$ be an entailment relation on a set $S$ generated by a set
  $\entails_0$ of axioms.
    If $\prox$ is an idempotent relation on $S$ such that
  \begin{enumerate}
    \item\label{prop:IndGenContEnt1} $C \prox_{U} A \entails_0 B \imp  \exists B' \in
      \Fin{S} \left(C \entails B' \prox_{L} B \right)$,
    \item\label{prop:IndGenContEnt2} $A \entails_0 B \prox_{L} C \imp \exists A' \in
      \Fin{S} \left(  A \prox_{U} A' \entails C\right)$,
  \end{enumerate}
  then $(S, \entails, \prox)$ is a strong continuous entailment relation.
\end{lemma}
\begin{proof}
  Let $\prox$ be an idempotent relation on $S$ satisfying
  \ref{prop:IndGenContEnt1} and \ref{prop:IndGenContEnt2}.
  We show only one direction of \eqref{eq:ContEnt},
  \[
     A \entails B
    \implies
     \forall C \prox_{U} A \exists B' \in \Fin{S}
     \left( C \entails B' \prox_{L}  B\right),
  \]
  by induction on the derivation of
  $A \entails B$.
  If $A \entails B$ is derived by ($\mathrm{R'}$), then the conclusion is trivial.
  Suppose that $A,A' \entails B$ is derived by 
  ($\mathrm{AxL}$). Then, there exists $C \in \Fin{S}$ such that $A
  \entails_{0} C$ and $\forall c \in C \left( 
  A',c \entails B \right)$. Let $D \prox_{U} A \cup A'$. Since $D \prox_U A$,
  there exists $C'\prox_L C$ such that $D \entails C'$ by
  \ref{prop:IndGenContEnt1}. By induction
  hypothesis, for each $c' \in C'$, there exists $B_{c'} \prox_L B$
  such that $D, c' \entails B_{c'}$. Put $B' = \bigcup_{c' \in C'} B_{c'}$.
  Then, by successive applications of ($\mathrm{T}$), we obtain $D \entails
  B'$.

  The other direction of \eqref{eq:ContEnt} follows from
  \ref{prop:IndGenContEnt2} and Lemma \ref{lem:IndGenEntDual}.
\end{proof}

As an application of generated strong continuous entailment
relations, we show that $\SProxLatP$ and $\jSPxLatP$ are equivalent.
Recall that the functor $F \colon \ContEnt \to \ProxLat$ restricts to
an equivalence of
$\SContEntP$ and $\SProxLatP$ (Theorem
\ref{thm:EuqivSContEntPAndSProxLatP}).  Composing $F$ with the
inclusion $\SProxLatP \hookrightarrow \jSPxLatP$, we get a full and
faithful functor $F' \colon \SContEntP \to \jSPxLatP$.
\begin{lemma}
  \label{lem:EquivAllProxLat}
  The functor  $F'$ is essentially surjective.
\end{lemma}
\begin{proof}
   Given a $\vee$-strong proximity lattice $(S,\prox)$,
  define an entailment relation $\entails^{\wedge}$ on
  $S$ by specifying its axioms as follows:
  \[
    A \entails^{\wedge} B \defeqiv 
    \exists C \in \Fin{S}
    \left( A \prox_{U} C \amp \medwedge C \leq \medvee B\right).
  \]
  Using Lemma \ref{prop:IndGenContEnt}, one can  show that 
  $(S, \entails^{\wedge}, \prox)$ is a strong continuous
  entailment relation.
  On the other hand, let $G(S, \prox) = (S, \entails, \ll)$ be the continuous
  entailment relation determined by the quasi-inverse $G$ of
  $F$ (see \eqref{eq:EntofDLat} and
  \eqref{eq:ProxLatToContEntll}). It suffices to show that $(S,
  \entails^{\wedge}, \prox)$ and  $(S, \entails, \ll)$ 
  are isomorphic as
  continuous entailment relations.
  By induction on $\entails^{\wedge}$, we see that
  \[
    A \ll_{\entails^{\wedge}}  B
    \iff
     \exists C \in \Fin{S} \left( A \prox_{U} C \amp \medwedge
    C \prox \medvee B\right).
  \]
  Then, it is straightforward to show that 
  $\ll \cdot \ll_{\entails^{\wedge}} 
  {=} \ll_{\entails^{\wedge}} 
  {=} \ll_{\entails^{\wedge}} \cdot \ll_{\entails^{\wedge}}
  $
  and 
  $ \ll_{\entails^{\wedge}}\cdot \ll 
  {=} \ll 
  {=} \ll \cdot \ll
  $.
  Thus $\ll$ and $\ll_{\entails^{\wedge}}$ are proximity maps
  $\ll \colon {(S, \entails, \ll)} \to (S, \entails^{\wedge},
  \ll_{\entails^{\wedge}})$ and $\ll_{\entails^{\wedge}} \colon (S,\entails^{\wedge}, \ll_{\entails^{\wedge}})
  \to (S, \entails, \ll)$
  and inverse to each other.
\end{proof}

\begin{theorem}
  \label{thm:EquivAllProxLat}
  The categories $\SContEntP$, $\ContEnt$, $\ProxLat$, $\jSPxLatP$,
  and $\SProxLatP$ are equivalent.
\end{theorem}
\begin{proof}
  By Lemma \ref{lem:EquivAllProxLat}, Theorem \ref{thm:EuqivSContEntPAndSProxLatP}, 
  and Theorem \ref{prop:EquivPxLatandjSPxLat}.
\end{proof}
Since every isomorphic proximity relation between $\vee$-strong
  proximity lattices are join-preserving, we also have the following
  by Theorem \ref{thm:EquiSContEntSProxLat}.
\begin{theorem}
  \label{thm:EquivAllSProxLat}
  The categories $\SContEnt$, $\jSPxLat$, and $\SProxLat$
  are equivalent.
  \qed
\end{theorem}

\section{De Groot duality}\label{sec:deGrootDuality}
In point-set topology, the de Groot dual of a stably compact space has the same
set of points equipped with the cocompact topology:
the topology generated by the complements of
compact saturated subsets of the original space.
By Hofmann--Mislove theorem, compact saturated subsets correspond
to Scott open filters, which are amenable to point-free treatment.
Thus, the \emph{de Groot dual} of a stably compact locale $X$ is
defined to be the locale whose frame is the Scott open filters on
$\Frame{X}$; see Escard\'o~\cite{escardoLawsonDual}.
We relate the de Groot duality to the
structural dualities of proximity lattices and continuous entailment relations.

\subsection{Duality of proximity lattices }\label{sec:DualitySProxLat}
\begin{definition}
  \label{def:DualProxLat}
  The \emph{dual} $\DualLat{S}$ of a distributive lattice $S =
  (S,0,\vee, 1, \wedge)$ is the distributive lattice $(S,1, \wedge,
  0,\vee)$ with the opposite order. 
  The \emph{dual} $\DeGroot{S}$ of a proximity lattice $S = (S,\prox)$ is
  the proximity lattice $(\DualLat{S}, \proxop)$.%
\end{definition}
Our aim is to give a localic account of \cite[Section
4]{JungSunderhaufDualtyCompactOpen}, which shows that
$\RIdeals{\DeGroot{S}}$ is isomorphic to the frame of Scott open
filters on $\RIdeals{S}$. 

\begin{definition}
  \label{def:ScottTop}
  Let $(S,\prox)$ be a proximity lattice.  Write $\Scott{S}$ for
  the locale whose models are rounded ideals of $S$, i.e., $\Scott{S}$
  is presented by a geometric theory $T_{\ScottFunc}$ over $S$ with the
  following axioms:
  \[
    \begin{gathered}
      \top \entails 0, \qquad\qquad \qquad
      a \wedge b \entails \left( a \vee b \right),
      \qquad\qquad\qquad
      a  \entails b \qquad (\text{if $b \leq a$}),\\
      a  \entails b \qquad (\text{if $b \prox a$}), \qquad\qquad
      a  \entails \bigvee_{b \succ a} b.
    \end{gathered}
  \]
\end{definition}
Let $\Ups{S}$ be the collection of rounded upper subsets of
$(S,\prox)$, 
i.e., those subset $U \subseteq S$ such that $a \in U \leftrightarrow
\exists b \prox a \left( b \in U \right)$.
Clearly, $\Ups{S}$ is closed under all joins, which are just unions.
Moreover, \ref{def:ProximityRelationVee} ensures that $\Ups{S}$ has
finite meets defined by
\begin{align*}
  1 &\defeql S = \upset_{\prox} 0, &
  U \wedge V &\defeql \bigcup_{a \in U, b \in V}
  \upset_{\prox} \left( a \vee b \right),
\end{align*}
where $\upset_{\prox} a \defeql \left\{ b \in S \mid b \succ a \right\}$.
These finite meets clearly distribute over all joins. Hence $\Ups{S}$ is a frame.

\begin{lemma}\label{lem:Scott}
The frames $\Frame{\Scott{S}}$ and $\Ups{S}$ are isomorphic.
\end{lemma}
\begin{proof}
  It is straightforward to show that a function $i_{\ScottFunc} \colon S \to
  \Ups{S}$ defined by
    $
    i_{\ScottFunc}(a) \defeql \upset_{\prox} a
    $
  is a universal interpretation of $T_{\ScottFunc}$.
\end{proof}
\begin{proposition}
  \label{prop:SigmaIsScott}
  The frame 
  $\Frame{\Scott{S}}$ is the Scott topology on $\RIdeals{S}$.
\end{proposition}
\begin{proof}
It is known that $\Ups{S}$ is the Scott topology on
$\RIdeals{S}$; see Vickers \cite[Lemma 2.11]{Infosys} or
Jung and S\"underhauf \cite[Lemma 14]{JungSunderhaufDualtyCompactOpen}.
Then, the claim follows from 
Lemma \ref{lem:Scott}.
\end{proof}

Scott open filters on a locale $X$ are models of 
the upper powerlocale of $X$, which is characterised by the following
universal property; see Vickers \cite{Vickers95constructivepoints}.
\begin{definition}\label{def:UpperPower}
  The \emph{upper powerlocale} of a locale $X$ is a locale
  $\Upper{X}$ together with a preframe homomorphism 
  $i_{U} \colon \Frame{X} \to \Frame{\Upper{X}}$ such that
  for any preframe homomorphism $f \colon \Frame{X} \to \Frame{Y}$
  to a locale $Y$, there exists a unique frame homomorphism
  $\overline{f} \colon \Frame{\Upper{X}} \to \Frame{Y}$ such that
  $\overline{f} \circ i_U = f$.
\end{definition}

\begin{proposition}
  \label{thm:ScottOfDualIsUpper}
  For any proximity lattice $S$,
  $\Scott{\DeGroot{S}}$ is the upper powerlocale of
  $\Spec{S}$.
\end{proposition}
\begin{proof}
  By Definition \ref{def:ScottTop},
  the locale $\Scott{\DeGroot{S}}$ is
  presented by a geometric theory $T$ over $S$ with 
  the following axioms:
  \begin{gather*}
    \top \entails 1,
   \qquad \qquad\qquad
    a \wedge b \entails (a \wedge b),
   \qquad \qquad\qquad
    a  \entails b \qquad (\text{if $a \leq b$}),\\
    a  \entails b \qquad (\text{if $a \prox b$}), \qquad\qquad
    a  \entails \bigvee_{b \prox a} b.
  \end{gather*}
  By the preframe version of Proposition
  \ref{prop:BijInterpretationScottCont} (see also Remark
  \ref{rem:DcpoPreframeInterpretation}), the universal interpretation
  $i_{T} \colon S \to \Frame{\Scott{\DeGroot{S}}}$ of $T$ in
  $\Scott{\DeGroot{S}}$ uniquely extends to a preframe homomorphism
  $i_{U} \colon \RIdeals{S} \to \Frame{\Scott{\DeGroot{S}}}$ via the
  function $i_{S} \colon S \to \RIdeals{S}$ defined by
  \eqref{eq:InjectionGenerator}.
  Then, it is straightforward to show that $i_{U}$ satisfies the universal
  property of the upper powerlocale of $\Spec{S}$.
\end{proof}

\begin{theorem}\label{thm:DeGrootDual}
  For any proximity lattice $S$, 
  the frame $\RIdeals{\DeGroot{S}}$ is isomorphic to the frame of
  Scott open filters on $\RIdeals{S}$.
  Thus, $\Spec{\DeGroot{S}}$ is the de Groot dual of $\Spec{S}$.
\end{theorem}
\begin{proof}
  Since $\RIdeals{\DeGroot{S}}$ is the
  collection of models of $\Scott{\DeGroot{S}}$, it is isomorphic to the
  frame of Scott open filters on $\RIdeals{S}$  by Proposition~\ref{thm:ScottOfDualIsUpper}
\end{proof}

We extend the duality to morphisms.
The following are obvious.
\begin{lemma} \label{lem:DualClosedRel}
  If $r \colon S \to S'$ is a proximity relation between
  proximity lattices,
  then the relational opposite $r^{-}\!$ is a proximity relation
  $r^{-} \colon \DeGroot{S'} \to \DeGroot{S}$.
  \qed
\end{lemma}
\begin{proposition}\label{prop:FirstDualtyThm}
  The assignment $r \mapsto r^{-}$ determines a dual isomorphism
  ${(\cdot)^{-} \colon \ProxLat \xrightarrow{\cong}
  \Opposite{\ProxLat}}$.
  \qed
\end{proposition}
All of the categories we have introduced so far ($\ProxLat$,
$\ContEnt$ etc.) are order-enriched categories,
where homsets are ordered by the set-theoretic inclusion.
Thus, the following notion applies.
\begin{definition}
  Let $\mathbb{C}$ be an order-enriched category.
  For morphisms $f \colon A \to B$ and $g \colon B
  \to A$, we say that $f$ is a \emph{left adjoint} to $g$ and $g$ is a
  \emph{right adjoint} to $f$ if $f \circ g \leq_{B} \id_{B}$ and
  $\id_{A} \leq_{A} g \circ f$, where $\leq_{A}$ and $\leq_{B}$ are
  the orders on $\Hom{\mathbb{C}}{A}{A}$ and $\Hom{\mathbb{C}}{B}{B}$
  respectively. In this case,  $(f,g)$ is called an adjoint
  pair of morphisms from $A$ to $B$.
\end{definition}

Let $\ProxLatPerfect$ be the subcategory of $\ProxLat$
where morphisms from $S$ to $S'$
are adjoint pairs of proximity relations from $S'$ to $S$.
The identity on $(S,\prox)$ is $(\prox,\prox)$, and the composition of
adjoint pairs $(s,r)$ and $(s',r')$ is $(s \circ s', r'
\circ r)$.

\begin{theorem}\label{thm:SecondDualtyThm}
  The assignment $(s,r) \colon S \to S' \mapsto (r^{-},s^{-}) \colon
  \DeGroot{S} \to
  \DeGroot{S'}$ determines an isomorphism
  $\DeGroot{(\cdot)} \colon \ProxLatPerfect \xrightarrow{\cong} \ProxLatPerfect$.
\end{theorem}
\begin{proof}
  Since the functor $(\cdot)^{-} \colon \ProxLat \to
  \Opposite{\ProxLat}$ preserves the order on morphisms, for any
  morphism $(s,r) \colon S \to S'$ in $\ProxLatPerfect$ (i.e.\
  an adjoint pair of proximity relations from $S'$ to $S$), the pair
  $(r^{-},s^{-})$ is an adjoint pair from $\DeGroot{S'}$ to
  $\DeGroot{S}$, i.e., a morphism $(r^{-},s^{-}) \colon
  \DeGroot{S} \to \DeGroot{S'}$ in $\ProxLatPerfect$.
\end{proof}
A locale map $f \colon X \to Y$ is \emph{perfect} if
the corresponding frame homomorphism $\Frame{f} \colon \Frame{Y} \to
\Frame{X}$ has a Scott continuous right adjoint $g \colon
\Frame{X} \to \Frame{Y}$. In this case, $g$ is necessarily a 
preframe homomorphism.
An adjoint pair $(s,r) \colon S \to S'$ of proximity relations in
$\ProxLatPerfect$ corresponds to a perfect map from $\Spec{S}$ to
$\Spec{S'}$.  Hence, Theorem \ref{thm:SecondDualtyThm} is a
manifestation of the de Groot duality of stably compact locales in the
setting of proximity lattice.

\subsection{Duality of continuous entailment relations}\label{sec:DualityContEnt}
We describe an analogous duality on the category $\ContEnt$, 
and relate it to the duality on $\ProxLat$ via the equivalence of the two
categories.
\begin{definition}
The \emph{dual} $\DualLat{\entails}$ of an entailment relation
$\entails$ on $S$ is the relational opposite: $A
\mathrel{\DualLat{\entails}} B \defeqiv B \entails A$. The \emph{dual}
$\DeGroot{S}$ of a continuous entailment relation $S = (S, \entails,
\ll)$ is the continuous entailment relation $(S, \DualLat{\entails},
\gg)$.
\end{definition}

If $r \colon S \to S'$ is a proximity map
between continuous entailment relations, then $r^{-}$ is a proximity map 
$r^{-} \colon \DeGroot{S'} \to \DeGroot{S}$. Then, the following is obvious.
\begin{proposition}
  \label{prop:DualityContEntClosed}
  The assignment $r \colon S \to S' \mapsto r^{-} \colon \DeGroot{S'}
  \to \DeGroot{S}$ determines a dual isomorphism $(\cdot)^{-} \colon
  \ContEnt \xrightarrow{\cong}  \Opposite{\ContEnt}$.
  \qed
\end{proposition}

\begin{theorem}
  \label{thm:ContEntCLosedtoSPxLatClosed}
  The equivalence 
  $F \colon \ContEnt \to \ProxLat$ 
  commutes with the dual isomorphisms $(\cdot)^{-}\!$ on
  $\ContEnt$ and $\ProxLat$ 
  up to natural isomorphism.
\end{theorem}
\begin{proof}
  For each continuous entailment relation 
  $(S,\entails, \ll)$, 
  define a relation
  $r_{S} \subseteq  \DualLat{D(S,\entails)} \times D(S,\DualLat{\entails})$ by
  \[
    \mathcal{U} \mathrel{r_{S}} \mathcal{V}
    \defeqiv \mathcal{U}^{*} \ApproxExt{\gg} \mathcal{V}.
  \]
  Since $\mathcal{U} \ApproxExt{\ll} \mathcal{V}
  \leftrightarrow \mathcal{V}^{*} \ApproxExt{\gg} \mathcal{U}^{*}$,
  one can easily show that $r_{S}$ is a proximity relation
  from $\DeGroot{F(S,\entails,\ll)}$
  to $F(\DeGroot{(S, \entails, \ll)})$
  with an inverse $t_{S}$ defined by
  \[
    \mathcal{V} \mathrel{t_{S}} \mathcal{U}
    \defeqiv \mathcal{U} \ApproxExt{\ll} \mathcal{V}^{*}.
  \]
  To see that $r_{S}$ is natural in $S$, 
  for any proximity map $r \colon (S,\entails, \ll) \to (S',
  \entails', \ll')$ and 
  for any $\mathcal{U} \in \Fin{\Fin{S'}}$
  and $\mathcal{V} \in \Fin{\Fin{S}}$, we have
  \begin{align*}
    \mathcal{U} \mathrel{(r_{S} \circ (\widetilde{r})^{-})} \mathcal{V} 
    &\iff
     \exists \mathcal{W} \in \Fin{\Fin{S}}
    \left(\mathcal{U} \mathrel{(\widetilde{r})^{-}}  \mathcal{W}
    \amp \mathcal{W} \mathrel{r_{S}}\mathcal{V}\right) \\
    &\iff
     \exists \mathcal{W} \in \Fin{\Fin{S}}
    \left(\mathcal{W} \mathrel{\widetilde{r}}  \mathcal{U}
    \amp \mathcal{V}^{*} \ApproxExt{\ll} \mathcal{W}\right) \\
    &\iff
    \mathcal{V}^{*} \mathrel{\widetilde{r}}  \mathcal{U} \\
    &\iff
     \exists \mathcal{W} \in \Fin{\Fin{S'}}
    \left(\mathcal{W}^{*} \ApproxExt{\ll'} \mathcal{U} 
    \amp \mathcal{V}^{*} \mathrel{\widetilde{r}} \mathcal{W}^{*}\right) \\
    &\iff
     \exists \mathcal{W} \in \Fin{\Fin{S'}}
    \left(\mathcal{U}  \mathrel{r_{S'}} \mathcal{W}
    \amp \mathcal{W} \mathrel{\widetilde{(r^{-})}}  \mathcal{V}\right) \\
    &\iff
    \mathcal{U} \mathrel{(\widetilde{(r^{-})} \circ r_{S'})}
    \mathcal{V}.
    \qedhere
  \end{align*}
\end{proof}

Let $\ContEntPerfect$ be the category of continuous entailment
relations and adjoint pairs of proximity maps which is defined
similarly as $\ProxLatPerfect$.
The following is analogous to Theorem \ref{thm:SecondDualtyThm}.
\begin{proposition}
  The assignment 
  $(s,r) \colon S \to S' \mapsto (r^{-},s^{-}) \colon \DeGroot{S} \to \DeGroot{S'}$
  determines an isomorphism $\DeGroot{(\cdot)} \colon
  \ContEntPerfect \xrightarrow{\cong} \ContEntPerfect$.
  \qed
\end{proposition}

\begin{theorem}
  \label{thm:DualityContEntProxLat}
  The category $\ContEntPerfect$ is equivalent to $\ProxLatPerfect$.  The
  equivalence commutes with the isomorphisms $\DeGroot{(\cdot)}\!$ on
  $\ContEntPerfect$ and $\ProxLatPerfect$ up to natural isomorphism.
\end{theorem}
\begin{proof}
Since the functor $F \colon \ContEnt \to \ProxLat$ preserves
the order on morphisms, it can be restricted to an equivalence
between $\ContEntPerfect$ and  $\ProxLatPerfect$.
The second statement follows from Theorem
\ref{thm:ContEntCLosedtoSPxLatClosed}.
\end{proof}

We introduce the notion of dual for strong continuous entailment
relations.
\begin{definition}
The \emph{dual} of a strong continuous entailment relation $(S,
\entails, \prox)$ is a strong continuous entailment relation
$(S,\DualLat{\entails}, \succ)$.
\end{definition}
Note that the inclusion $\SContEntP \hookrightarrow \ContEnt$ 
commutes with the dualities on both categories.
Since strong proximity lattices are closed under the duality in the sense
of Definition \ref{def:DualProxLat}, Theorem
\ref{thm:DualityContEntProxLat} restricts to the full subcategories 
$\SContEntPerfect$
and $\SProxLatPerfect$
of $\ContEntPerfect$ and $\ProxLatPerfect$, respectively, which
consist of strong
continuous entailment relations and strong proximity lattices.

\section{Applications of entailment relations}\label{sec:ExampledeGrootDuality}
We present a number of constructions on stably compact locales in
the setting of strong proximity lattices and
strong continuous entailment relations and analyse their de Groot duals.
For the sake of simplicity, we prefer to work with \emph{strong} 
proximity lattices rather proximity lattices because the geometric
theories of the locales represented by the former are simpler and
easier to work with.

Our main tool is the following observation, together with Lemma
\ref{prop:IndGenContEnt}.
\begin{lemma}\label{lem:GenEntOp}
  If $(S, \entails)$ is an entailment relation generated by a set
  $\entails_{0}$ of axioms, then the dual $\DualLat{\entails}$ is generated by
  $\DualLat{\entails_{0}} \defeql \left\{ (B,A) \mid A \entails_0 B \right\}$.
\end{lemma}
\begin{proof}
  Immediate from the structural symmetry of entailment relations.
\end{proof}

\subsection{Powerlocales}\label{subsubsec:PowerLoc}
We deal with the lower, upper, and Vietoris powerlocales and 
consider their interactions with the construction
$\Scott{S}$, the locale whose frame is the Scott topology on
$\RIdeals{S}$.
For the localic account of powerlocales, the reader is referred to
Vickers~\cite{Vickers95constructivepoints,DoublePowLocExp}. 

\subsubsection{Lower and upper powerlocales}
\begin{lemma}\label{lem:PresentationUpperPowerLoc}
Let $(S,\prox)$ be a strong proximity lattice.
\begin{enumerate}
  \item\label{lem:PresentationUpperPowerLoc1} The locale $\Scott{S}$ is presented by a
    strong continuous entailment relation $\ScottFunc(S) = (S, \entails_{\Sigma}, \proxop)$
    where $\entails_{\Sigma}$ is generated by the following axioms:
  \begin{gather*}
    \label{eq:Scott}
    \mathrel{\entails_{\Sigma} 0}
    \qquad\qquad
    a, b \mathrel{\entails_{\Sigma}} a \vee b
    \qquad\qquad
    a  \mathrel{\entails_{\Sigma}} b \quad (\text{if $b \leq a$})
  \end{gather*}

  \item\label{lem:PresentationUpperPowerLoc2} The upper powerlocale of $\Spec{S}$ is presented by a
    strong continuous entailment relation $\Upper{S} = (S, \entails_U, \prox)$ where
    $\entails_U$ is generated by the following axioms:
  \begin{gather*}
    \label{eq:Upper}
    \mathrel{\entails_U 1}
    \qquad\qquad
    a, b \mathrel{\entails_U} a \wedge b
    \qquad\qquad
    a  \mathrel{\entails_U} b \quad (\text{if $a \leq b$})
  \end{gather*}
  \end{enumerate}
  In particular, (the locale whose frame is) the Scott topology and the upper powerlocale of a
  stably compact locale are stably compact.
\end{lemma}
\begin{proof}
  It is straightforward to check that $\ScottFunc(S)$
  and $\Upper{S}$ satisfy the condition in Lemma
  \ref{prop:IndGenContEnt}. Then, item \ref{lem:PresentationUpperPowerLoc1} is immediate from Definition
  \ref{def:ScottTop}, while item \ref{lem:PresentationUpperPowerLoc2} follows from
  Proposition~\ref{thm:ScottOfDualIsUpper}.
\end{proof}
Note that
  $A \entails_{\ScottFunc} B 
  \leftrightarrow \exists b \in B \left( b \leq
  \medvee A \right)$ and $A \entails_{U} B 
  \leftrightarrow \exists b \in B \left( \medwedge A
  \leq b \right)$.
The constructions $\ScottFunc(S)$ and $\Upper{S}$ extend to functors
$\ScottFunc \colon \Opposite{\SProxLat} \to \SContEnt$ and $\UpperFunc
\colon \SProxLat \to \SContEnt$, which send each join-preserving proximity relation 
$r \colon (S, \prox) \to (S',\prox')$
to join-preserving proximity maps $\ScottFunc(r) \colon \ScottFunc(S') \to \ScottFunc(S)$ and 
$\Upper{r} \colon \Upper{S} \to \Upper{S'}$ defined by
\[
  \begin{aligned} 
    A \mathrel{\ScottFunc(r)} B \defeqiv  \exists b \in B \left( b
    \mathrel{r} \medvee A\right), \\
    A \mathrel{\Upper{r}} B \defeqiv  \exists b \in B
    \left( \medwedge A \mathrel{r} b \right).
  \end{aligned}
\]

The notion of lower powerlocale is the dual of that
of upper powerlocale.
\begin{definition}\label{def:LowerPower}
  The \emph{lower powerlocale} of a locale $X$ is a locale
  $\Lower{X}$ together with a suplattice  homomorphism 
  $i_{L} \colon \Frame{X} \to \Frame{\Lower{X}}$ such that
  for any suplattice homomorphism $f \colon \Frame{X} \to \Frame{Y}$
  to a locale $Y$, there exists a unique frame homomorphism
  $\overline{f} \colon \Frame{\Lower{X}} \to \Frame{Y}$ such that
  $\overline{f} \circ i_L = f$.
\end{definition}
\begin{lemma}\label{lem:PresentationLowerPowerLoc}
  For any strong proximity lattice $(S,\prox)$, the lower powerlocale
  of $\Spec{S}$ is presented by a strong continuous entailment relation
  $\Lower{S} = {(S, \entails_L, \prox)}$ where $\entails_L$ is generated by 
  the following axioms:
  \begin{gather*}
    \label{eq:Lower}
    0 \mathrel{\entails_L}
    \qquad\qquad\qquad
    a \vee b \mathrel{\entails_L} a, b
    \qquad\qquad\qquad
    a  \mathrel{\entails_L} b \quad (\text{if $a \leq b$})
  \end{gather*}
  In particular, the lower powerlocale of a stably compact locale is stably
  compact.
\end{lemma}
\begin{proof}
  Immediate from the suplattice version of Proposition \ref{prop:BijInterpretationScottCont}.
\end{proof}
Note that
  $
  A \entails_{L} B \leftrightarrow \exists a \in A \left( a \leq
  \medvee B \right).
  $
The construction $\Lower{S}$ extends to 
a functor
$\LowerFunc \colon \SProxLat \to \SContEnt$, which sends each
join-preserving proximity relation $r \colon (S, \prox) \to (S',\prox')$
to a join-preserving proximity map $\Lower{r} \colon \Lower{S} \to
\Lower{S'}$ defined by
\[
  A \mathrel{\Lower{r}} B 
  \defeqiv 
  \exists a \in A \left( a  \mathrel{r} \medvee B\right).
\]
\begin{theorem}
  \label{thm:DualLowerUppwer}
  For any strong proximity lattice $S$, we have
  \begin{enumerate}
    \item\label{thm:DualLowerUppwer1}
      $\DeGroot{\Upper{S}} \cong \Lower{\DeGroot{S}}$
      and 
      $\DeGroot{\Lower{S}} \cong \Upper{\DeGroot{S}}$,

    \item\label{thm:DualLowerUppwer2}
      $\ScottFunc(\DeGroot{S}) \cong \Upper{S}$
      and
      $\DeGroot{\ScottFunc(S)} \cong \Lower{S}$.
  \end{enumerate}
\end{theorem}
\begin{proof}
  Immediate from
  Lemma \ref{lem:GenEntOp}, 
  Lemma \ref{lem:PresentationUpperPowerLoc}, and
  Lemma \ref{lem:PresentationLowerPowerLoc}.
\end{proof}
  Item \ref{thm:DualLowerUppwer1} of Theorem \ref{thm:DualLowerUppwer}
  is known: Vickers gave a localic proof using entailment systems
  \cite[Theorem 54]{VickersEntailmentSystem}, and Goubault-Larrecq
  proved the corresponding result for stably compact
  spaces~\cite[Theorem 3.1]{GoubaultLarrecq-ModelofChoice}.  It is
  notable, however, that our proof is a simple analysis of axioms of
  entailment relations.

In the following, compositions such as
$\Upper{\ScottFunc(S)}$ should be read as $\Upper{F(\ScottFunc(S))}$, where
$F \colon \SContEnt \to \SProxLat$ is the functor establishing the equivalence of
the two categories (see Remark~\ref{rem:EquivSContEntSProxLat}).
\begin{proposition}
  \label{prop:DualLowerUppwer}
  For any strong proximity lattice $S$, we have
  \begin{enumerate}
    \item\label{prop:DualLowerUppwer1} $\Upper{\ScottFunc(S)} \cong \ScottFunc(\Lower{S})$,
    \item\label{prop:DualLowerUppwer2} $\Upper{\Lower{S}} \cong \ScottFunc(\ScottFunc(S))$.
  \end{enumerate}
\end{proposition}
\begin{proof}
  By Theorem \ref{thm:ContEntCLosedtoSPxLatClosed} and
  item \ref{thm:DualLowerUppwer2} of Theorem \ref{thm:DualLowerUppwer},
  we have
  \[
    \Upper{\ScottFunc(S)} 
    \cong
    \Upper{\DeGroot{\Lower{S}}}
    \cong
    \ScottFunc(\Lower{S}).
  \]
  The proof of item \ref{prop:DualLowerUppwer2} is similar.
\end{proof}

\subsubsection{Double powerlocale}
\begin{definition}\label{def:DoublePower}
  The \emph{double powerlocale} of a locale $X$ is a locale
  $\Double{X}$ together with a Scott continuous function
  $i_{D} \colon \Frame{X} \to \Frame{\Double{X}}$ such that
  for any Scott continuous function $f \colon \Frame{X} \to \Frame{Y}$
  to a locale $Y$, there exists a unique frame homomorphism
  $\overline{f} \colon \Frame{\Double{X}} \to \Frame{Y}$ such that
  $\overline{f} \circ i_{D} = f$.
\end{definition}
\begin{lemma}
  \label{lem:DoublePowerSContEnt}
  For any strong proximity lattice $(S, \prox)$,
  the double power locale of $\Spec{S}$
  can be presented by a strong continuous entailment relation 
  $\Double{S} = (S, \entails_{D}, \prox)$, where $\entails_{D}$ is
  generated by the following  axioms:
  \begin{align*}
    a \entails_{D} b \qquad (a \leq b).
  \end{align*}
\end{lemma}
\begin{proof}
  By the dcpo version of Proposition \ref{prop:BijInterpretationScottCont}.
\end{proof}
Note that
  $
  A \entails_{D} B \leftrightarrow \exists a \in A \exists b \in B \left( a \leq b
  \right).
  $
The construction $(S, \prox) \mapsto \left( S, \entails_{D}, \prox \right)$
extends to a functor $\DoubleFunc \colon \SProxLat \to
\SContEnt$, which sends each join-preserving proximity relation
$r \colon (S, \prox) \to (S', \prox')$ to
a join-preserving proximity map  
$\Double{r} \colon (S, \entails_{D}, \prox) \to (S, \entails_{D}', \prox')$
defined by
\[
  A \mathrel{\Double{r}} B \defeqiv \exists a \in A \exists b \in B
  \left( a \mathrel{r} b \right).
\]
\begin{proposition}
  \label{prop:DoubleSelfDual}
  For any strong proximity lattice $S$, we have 
  \[
    \DeGroot{\Double{S}} \cong \Double{\DeGroot{S}}.
  \]
\end{proposition}
\begin{proof}
  Immediate from Lemma \ref{lem:GenEntOp} and Lemma
  \ref{lem:DoublePowerSContEnt}. 
\end{proof}

We prove some well-known characterisations of the double
powerlocale (see Proposition~\ref{prop:PreframeCoverage} and
Proposition~\ref{prop:DoubleScott}).
To this end, we begin with the construction of the lower powerlocale of a 
strong continuous entailment relation.

Given a strong continuous entailment relation $(S, \entails, \prox)$
define an entailment relation $\entails^{L}$ on
$\Fin{S}$ by the following axioms:
\begin{align*}
  A &\entails^{L} A_{0}, \dots, A_{n-1} &&
  (\text{if $\forall B \in \left\{ A_{i} \mid i < n \right\}^{*}
A \entails B$})
\end{align*}
Define an idempotent relation $\prox^{L}$ on $\Fin{S}$ by
\[
  A \prox^{L} B \defeqiv A \prox_{U} B.
\]
Then $(\Fin{S}, \entails^L, \prox^{L})$ is a strong continuous
entailment relation by Lemma \ref{prop:IndGenContEnt}.
\begin{lemma}
  \label{lem:LowerPowEntSys}
  For any $\mathcal{U}, \mathcal{V} \in \Fin{\Fin{S}}$, we have
  \[
    \mathcal{U} \ll_{\entails^L} \mathcal{V}
    \iff
    \exists A \in \mathcal{U} \forall B \in \mathcal{V}^{*}
    \left( A \ll_{\entails} B \right).
  \]
\end{lemma}
\begin{proof}
  By induction on $\entails^L$, one can show that 
  \begin{equation}
    \label{eq:vdashL}
    \mathcal{U} \entails^{L} \mathcal{V}
    \defeqiv \exists  A \in \mathcal{U} \forall B \in
    \mathcal{V}^{*}\left(  A \entails  B\right).
  \end{equation}
  Then, the direction $\Rightarrow$ is obvious from \eqref{eq:vdashL}.
  Conversely, suppose that 
  there exists 
  $A \in
  \mathcal{U}$ such that $\forall B \in \mathcal{V}^{*} \left( A
  \ll_{\entails} B \right)$.
  Then, for each $B \in \mathcal{V}^{*}$, there exists $C_{B}$ such
  that $A \prox_{U} C_{B} \entails B$. Put $C =
  \bigcup_{B \in \mathcal{U}^{*}} C_{B}$. Then $A \prox_{U} C$
  and $C \entails B$ for all $B \in \mathcal{V}^{*}$. Hence 
  $\mathcal{U} \mathrel{(\prox^{L})_{U}} \left\{ C \right\} \entails^{L} \mathcal{V}$
  and so $\mathcal{U} \ll_{\entails^{L}} \mathcal{V}$.
\end{proof}
\begin{lemma}
  \label{lem:LowerPowerSContEnt}
  The strong continuous entailment relation $(\Fin{S}, \entails^L, \prox^{L})$
  presents the lower powerlocale of $\Spec{F(S, \entails, \prox)}$.
\end{lemma}
\begin{proof}
From the characterisation of $\ll_{\entails^L}$ in Lemma
\ref{lem:LowerPowEntSys}, the entailment system $(\Fin{S},
\ll_{\entails^{L}})$ coincides with the construction of the lower
powerlocale of the entailment system  ${(S, \ll_{\entails})}$ in
Vickers \cite[Theorem 53]{VickersEntailmentSystem}. 
\end{proof}
\begin{corollary}
  \label{lem:UpperPowerSContEnt}
The upper powerlocale of $\Spec{F(S,
\entails, \prox)}$ can be presented by a strong continuous entailment relation
$(\Fin{S}, \entails^{U}, \prox^{U})$ defined by
\[
  \begin{aligned}
    \mathcal{U} \entails^{U} \mathcal{V}
    &\defeqiv \exists B \in \mathcal{V} \forall A \in
    \mathcal{U}^{*} \left( A \entails B \right), \\
    A \prox^{U} B 
    &\defeqiv
    A \prox_{L} B.
  \end{aligned}
\]
\end{corollary}
\begin{proof}
We have $\DeGroot{\Lower{\DeGroot{S}}} \cong \Upper{S}$ by item
\ref{thm:DualLowerUppwer2} of Theorem \ref{thm:DualLowerUppwer}. The
corollary follows by unfolding the definition of
$\DeGroot{\Lower{\DeGroot{S}}}$ using Lemma \ref{lem:LowerPowEntSys}.
\end{proof}

For any strong proximity lattice $(S, \prox)$, define a preorder
$\leq^{\vee}$ on $\Fin{S}$ by 
\[
  A \leq^{\vee} B \defeqiv \medvee A \leq \medvee B.
\]
Let $S^{\vee}$ be the set $\Fin{S}$ equipped with the equality
determined by $\leq^{\vee}$.
Define
a lattice structure $S^{\vee} \defeql (S^{\vee}, 0^{\vee},
\vee^{\vee}, 1^{\vee}, \wedge^{\vee})$
as in \eqref{eq:LatticeStructVee} and an idempotent relation $\prox^{\vee}$ on $S^{\vee}$ by 
\[
  A \prox^{\vee} B \defeqiv \medvee A \prox \medvee B.
\]
Then, $(S^{\vee}, \prox^{\vee})$
is a strong proximity lattice, which is isomorphic to $(S, \prox)$ via
proximity relations $r \colon (S, \prox) \to (S^{\vee}, \prox^{\vee})$
and $s \colon (S^{\vee}, \prox^{\vee}) \to (S, \prox)$ defined
by
\[
    a \mathrel{r} A \defeqiv a \prox \medvee A, \qquad\quad
    A \mathrel{s} a \defeqiv \medvee A \prox a.
\]

The following is a special case of Vickers
\cite{DoublePowLocExp}, which holds for more general context of locally compact
locale.\footnote{In \cite{DoublePowLocExp},
$\ScottFunc(S)$ is expressed as the exponential over the Sierpinski
locale.} 
\begin{proposition}
  \label{prop:PreframeCoverage}
  For any strong proximity lattice $(S, \prox)$, we have
  \[
    \Double{S} \cong \ScottFunc(\ScottFunc(S)).
  \]
\end{proposition}
\begin{proof}
  By item \ref{prop:DualLowerUppwer2} of Proposition
  \ref{prop:DualLowerUppwer}, it suffices to show that
  $\Double{S} \cong \Upper{\Lower{S}}$.
  By Lemma \ref{lem:PresentationLowerPowerLoc} and Corollary
  \ref{lem:UpperPowerSContEnt}, the locale $\Upper{\Lower{S}}$ is
  presented by a strong continuous entailment relation $(\Fin{S},
  \entails^{\mathit{UL}}, \prox^{\mathit{UL}})$ on $\Fin{S}$ defined by
    \begin{align*}
      \mathcal{U} \entails^{\mathit{UL}} \mathcal{V}
      &\defeqiv
      \exists B \in \mathcal{V} \forall A \in
      \mathcal{U}^{*} \exists a \in A \left( a \leq \medvee B
      \right)\\
      &\iff
      \exists B \in \mathcal{V} \exists A \in
      \mathcal{U}  \left( \medvee A \leq \medvee B \right), \\
      A \prox^{\mathit{UL}} B &\defeqiv A \prox_{L} B. \\
  \shortintertext{Thus}
      \mathcal{U} \ll_{\entails^{\mathit{UL}}} \mathcal{V}
      &\iff \exists A \in \mathcal{U} \exists B \exists C \in
      \mathcal{V} 
      \left( A \prox_{L} B \amp \medvee B \leq \medvee C \right) \\
      &\iff \exists A \in \mathcal{U} \exists C \in \mathcal{V}
      \left(\medvee A \prox \medvee C \right).
  \end{align*}
  As for the strong proximity lattice $(S^{\vee},
  \prox^{\vee})$ defined above,
  its double powerlocale
  $\Double{S^{\vee}} = (S^{\vee}, \entails^{\vee}, \prox^{\vee})$
  characterised in Lemma \ref{lem:DoublePowerSContEnt}
  satisfies
  \[
    \mathcal{U} \ll_{\entails^{\vee}} \mathcal{V}
    \iff 
    \exists A \in \mathcal{U} \exists B \in \mathcal{V}
    \left( \medvee A \prox \medvee B \right).
  \]
  Clearly, the entailment systems $(S^{\vee}, \ll_{\entails^{\vee}})$
  and $(\Fin{S}, \ll_{\entails^{\mathit{UL}}})$ are isomorphic.
  Since $\DoubleFunc$ is functorial and the embedding 
  $\SContEntP \hookrightarrow \Ent$ is faithful, we have
   $\Double{S} \cong \Double{S^{\vee}} \cong
   \Upper{\Lower{S}}$.
\end{proof}

\begin{corollary}
  \label{cor:DoubleScott}
  For any strong proximity lattice $S$,
  we have 
  \[
    \DeGroot{\ScottFunc(\ScottFunc(S))} \cong \ScottFunc(\ScottFunc(\DeGroot{S})).
  \]
\end{corollary}
\begin{proof}
  By Proposition \ref{prop:DoubleSelfDual} and  Proposition \ref{prop:PreframeCoverage}.
\end{proof}

\begin{proposition}
  \label{prop:DoubleScott}
  For any strong proximity lattice $S$,
  we have 
  \begin{enumerate}
    \item\label{prop:DoubleScott1} $\Lower{\ScottFunc(S)} \cong \ScottFunc(\Upper{S})$,
    \item\label{prop:DoubleScott2} $\Lower{\Upper{S}} \cong
      \Double{S}$.
  \end{enumerate}
\end{proposition}
\begin{proof}
  \noindent \ref{prop:DoubleScott1}.  By Corollary \ref{cor:DoubleScott}, and item 
  \ref{thm:DualLowerUppwer2} of Theorem \ref{thm:DualLowerUppwer}. 
  
  \noindent \ref{prop:DoubleScott2}.  Apply item
  \ref{prop:DoubleScott1} to $\DeGroot{S}$ and use Theorem
  \ref{thm:DualLowerUppwer} and Proposition
  \ref{prop:PreframeCoverage}.
\end{proof}
Item \ref{prop:DualLowerUppwer2} of Proposition \ref{prop:DualLowerUppwer} and Proposition
\ref{prop:DoubleScott} are known for locally compact locales; see
Vickers~\cite{DoublePowLocExp}.
Item \ref{prop:DualLowerUppwer1} of Proposition \ref{prop:DualLowerUppwer} and Proposition
\ref{prop:DoubleScott} say that the locales of the form
$\ScottFunc(S)$ for a stably compact $S$ is closed under the lower and
upper powerlocales. Moreover, the lower and upper
powerlocales of $\ScottFunc(S)$ are obtained by the upper and the
lower powerlocales of $S$, respectively.
\subsubsection{Vietoris powerlocale}
\begin{definition}\label{def:Vietoris}
  Let $(S, \prox)$ be a strong
  proximity lattice. 
  The \emph{Vietoris powerlocale} of $\Spec{S}$ is presented
  by 
  a strong continuous entailment relation $\Vietoris{S} = (S_{V}, \entails_V,
  \prox_{V})$ on the set 
    $
    S_{V} \defeql \left\{ \Diamond a \mid a \in S \right\} \cup 
    \left\{ \Box a \mid a \in S \right\},
    $
  where $\entails_V$ is
  generated by the following axioms:
  \begin{gather*}
    \begin{aligned}
    \Diamond 0 &\entails_V 
    &\qquad
    \Diamond (a \vee b) &\entails_V \Diamond a, \Diamond b
    &\qquad
    \Diamond a &\entails_V \Diamond b \quad (\text{if $a \leq b$}) \\
     &\entails_V \Box 1
    &\qquad
    \Box a, \Box b &\entails_V \Box (a \wedge b) 
    &\qquad
    \Box a &\entails_V \Box b \quad (\text{if $a \leq b$})\\
    \end{aligned}\\
    \begin{aligned}
    \Box a, \Diamond b &\entails_V \Diamond (a \wedge b) &\qquad
    \Box (a \vee b) &\entails_V \Box a, \Diamond b
    \end{aligned}
  \end{gather*}
  The idempotent relation $\prox_{V}$ is defined by
  \begin{align*}
    \Diamond a \prox_V \Diamond b \defeqiv a \prox b,
    &&
    \Box a \prox_V \Box b \defeqiv  a \prox b.
  \end{align*}
\end{definition}
One can easily verify that $\Vietoris{S}$ satisfies
the condition in Lemma~\ref{prop:IndGenContEnt}.  Moreover, it is
straightforward to show that the locale presented by
$\Vietoris{S}$ is isomorphic to the Vietoris powerlocale of
$\Spec{S}$; see Johnstone \cite{JohnstoneVietorisLocSemiLat} for the
construction of Vietoris powerlocales.
Thus, the Vietoris powerlocale of a stably compact locale is stably compact.

The construction $\Vietoris{S}$ extends to 
a functor $\VietorisFunc \colon \SProxLat \to \SContEnt$, which sends each
join-preserving proximity relation $r \colon (S, \prox) \to (S',\prox')$
to a join-preserving proximity map $\Vietoris{r} \colon \Vietoris{S} \to
\Vietoris{S'}$ defined by
\[
  \Diamond A \Box B \mathrel{\Vietoris{r}} \Diamond C \Box  D
  \defeqiv
   \exists a \in A \left( a \wedge \medwedge B  \mathrel{r}
  \medvee C\right) 
  \;\text{or}\; 
   \exists d \in D \left( \medwedge B  \mathrel{r}
  d \vee \medvee C\right),
\]
where $\Diamond A \Box B
\defeql \left\{ \Diamond a \mid a \in A \right\} \cup \left\{ \Box b
\mid b \in B \right\}$ for each $A,B \in \Fin{S}$.
\begin{theorem}
  \label{thm:Vietoris}
  For any strong proximity lattice $S$,
  we have 
  \[
    \DeGroot{\Vietoris{S}} \cong \Vietoris{\DeGroot{S}}.
  \]
\end{theorem}
\begin{proof}
  $\DeGroot{\Vietoris{S}}$ and $\Vietoris{\DeGroot{S}}$ are
  identical except that $\Diamond$ and
  $\Box$ are swapped.
\end{proof}
Goubault-Larrecq proved the result corresponding to Theorem
\ref{thm:Vietoris} for stably compact
spaces using $\mathbf{A}$-valuations \cite[Corollary 5.24]{GoubaultLarrecq-ModelofChoice}.

\subsection{Patch topologies}\label{subsec:Patch}
Coquand and Zhang~\cite{CoquandZhangPrediativePatch} gave a
construction of patch topologies for entailment relations with the
interpolation property.
The same construction carries over to the setting of strong continuous
entailment relation.
\begin{definition}[{Coquand and Zhang \cite[Section
  4]{CoquandZhangPrediativePatch}}]\label{def:Patch}
  Given a strong continuous entailment relation $(S, \entails,\prox)$,
  the \emph{patch topology} of $S$ is a
  strong continuous entailment relation $\Patch{S} = (S_{P}, \entails_P,
  \prox_{P})$ on the set
    $
    S_{P} \defeql S \cup \left\{ \dual{a} \mid a \in S \right\},
    $
  where $\entails_P$ is generated by the following axioms:
  \begin{align}
   \label{Patch:Axentails}
   A &\mathrel{\entails_P} B &&(\text{if $A \entails B$})\\
   \label{Patch:AxD} a,\dual{b} &\mathrel{\entails_{P}}  &&(\text{if $a \prox b$}) \\
   \label{Patch:AxLoc} &\mathrel{\entails_{P}} \dual{a}, b &&(\text{if $a \prox b$})
  \end{align}
  The idempotent relation $\prox_{P}$ is defined by
  \begin{align*}
    a \prox_P b \defeqiv  a \prox b,
    &&
    \dual{a} \prox_P \dual{b} \defeqiv  b \prox a.
  \end{align*}
\end{definition}
Let $\PatchP{S} = (S_{P}, \entails_{P}', \prox_{P})$ be
the strong continuous entailment relation which is obtained from $\Patch{S}$
by adjoining the following axioms:
\begin{equation}
  \label{Patch:Axentailsop}  \dual{B} \mathrel{\entails_P'}
  \dual{A} \qquad (\text{if $A \entails B$})
\end{equation}
where $\dual{A} \defeql \left\{ \dual{a} \mid a \in A \right\}$ for each
$A \in \Fin{S}$.
\begin{proposition}
  \label{prop:PatchCut}
  For any strong continuous entailment relation $(S,\entails, \prox)$,
  we have 
  \[
   \PatchP{S} \cong \Patch{S}.
  \]
\end{proposition}
\begin{proof}
  We show that the entailment systems associated with
  $\PatchP{S}$ and $\Patch{S}$ coincide, i.e.,
  $\ll_{\entails_P} {=} \ll_{\entails_{P}'}$. 
  Since $\entails_{P}'$
  is generated by the extra axioms, we have
  $\ll_{\entails_P} \mathop{\subseteq} \ll_{\entails_{P}'}$.
  To prove the converse inclusion, it suffices to
  show that
  \[
    X \entails_{P}' Y \implies \forall Z \left(\prox_{P} \right)_{U} X \left( Z
    \ll_{\entails_{P}} Y \right).
  \]
  This is
  proved by induction on the derivation of $X \entails_{P}' Y$ (see Lemma~\ref{lem:IndGenEnt}).
  The case $(\mathrm{R'})$ is obvious,
  so it suffices to check the case $(\mathrm{AxL})$ for each axiom of
  $\Patch{S'}$.  We only deal with 
  \eqref{Patch:Axentailsop}.
  Suppose that $\DualLat{B}, X \entails_{P}' Y$ is derived
  from $\DualLat{B} \entails_{P}' \DualLat{A}$ and $\forall a \in A
  \left( X, \DualLat{a} \entails_{P}' Y \right)$  where $A \entails B$.
  Let $Z \left( \prox_{P} \right)_{U} \DualLat{B},X$. 
  Then, there exists $C \in \Fin{S}$ such that
  $\DualLat{C} \subseteq Z$ and $B \prox_L C$, and so 
  there exists $C'$ such that $B \prox_{L} C' \prox_{L} C$.
  Thus, there exists
  $D$ such that $A \prox_{U} D \entails C'$ by \eqref{eq:ContEnt}.
  For each $d \in D$, there exist $a \in A$ and $d'$ such that
  $a \prox d' \prox d$. Then, $Z, \DualLat{d'}
  \left(\prox_{P}\right)_{U} X, \DualLat{a}$ so $Z, \DualLat{d'}
  \ll_{\entails_{P}} Y$ by induction hypothesis.
  Thus, for each $d \in D$, there exist $d' \prox d$
  and $W_{d} \in \Fin{S}$ such that $Z, \DualLat{d'}
  \entails_{P} W_{d} \left(\prox_{P}\right)_{L} Y$ so that
  $Z \entails_P W_{d}, d$ by \eqref{Patch:AxLoc}. Since $D \entails C'
  \prox_{L} C$, we get $Z, \DualLat{C}\entails_{P}
  \bigcup_{d \in D} W_{d}$ by \eqref{Patch:Axentails}
  and successive applications of
  $(\mathrm{T})$ and \eqref{Patch:AxD}.
  Hence, $Z = Z \cup \DualLat{C} \ll_{\entails_{P}} Y$.
\end{proof}
In terms of $\SContEnt$, we have proximity maps 
$r \colon \Patch{S} \to \PatchP{S}$ and 
$s \colon \PatchP{S} \to \Patch{S}$ defined by
\begin{align*}
  A \mathrel{r} B &\defeqiv A \ll_{\entails_P} B, &
  A \mathrel{s} B &\defeqiv A \ll_{\entails_P'} B,
\end{align*}
which are inverse to each other.
\begin{theorem}\label{theorem:patch}
  For any strong continuous entailment relation $S$,
  we have 
  \[
    \Patch{S} \cong \Patch{\DeGroot{S}}.
  \]
\end{theorem}
\begin{proof}
  By Proposition \ref{prop:PatchCut}, we may identify
  $\Patch{S}$ with $\PatchP{S}$.
   Then,
   we have $\PatchP{S} \cong \PatchP{\DeGroot{S}}$
   by exchanging the roles of $a$ and $\dual{a}$.
\end{proof}
\begin{remark}
 Combining 
 $\PatchFunc$ and the equivalence between
 $\SProxLatPerfect$ and $\SContEntPerfect$,
 we get a functor $\PatchFunc \colon \SProxLatPerfect \to
 \SContEntPerfect$, which sends an adjoint pair
 $(s,r)$ of proximity relations $r
 \colon (S, \prox) \to (S', \prox')$ and $s \colon {(S', \prox')} \to
 {(S, \prox)}$ to an adjoint pair $(\mathfrak{P}(s),\mathfrak{P}(r))$
 of proximity maps $\mathfrak{P}(r) \colon {\Patch{G(S)}} \to
 \Patch{G(S')}$
 and $\mathfrak{P}(s) \colon \Patch{G(S')} \to \Patch{G(S)}$
 defined by
 \begin{align*}
   A,\dual{B} \mathrel{\mathfrak{P}(r)} C,\dual{D}
   &\defeqiv
    \exists a, b \in S \left( \medwedge A \wedge b \prox \medvee B \vee a
   \amp a \mathrel{r} \medvee C \amp \medwedge D \mathrel{s} b\right),\\
   C,\dual{D} \mathrel{\mathfrak{P}(s)} A,\dual{B}
   &\defeqiv
    \exists a, b \in S \left( \medwedge B \wedge a \prox \medvee A \vee b
   \amp b \mathrel{r} \medvee D \amp \medwedge C \mathrel{s} a\right).
 \end{align*}
\end{remark}

\subsection{Space of valuations}\label{subsec:Valuations}
The space of valuations is a localic analogue of the \emph{probabilistic
power domain} by Jones and Plotkin \cite{JonesProbablistic,JonesPlotkinProbabilisticPD}.
We first recall several notions of real
numbers which are needed for its definition.
  \begin{enumerate}
    \item A \emph{lower real} is a rounded downward closed subset of
      rationals $\Rat$. 

    \item An \emph{upper real} is a rounded upward closed subset of
      $\Rat$. 

    \item A \emph{Dedekind real} is a disjoint pair $(L,U)$ of an
      inhabited lower real $L$ and an inhabited upper real $U$ which
      is located: $p < q$ implies $p \in L$ or
      $q \in U$.  
  \end{enumerate}
      Let $\overrightarrow{[0,\infty]}$ and 
      $\overleftarrow{[0,\infty]}$ denote
      the lower and the upper reals greater than $0$
      respectively (including infinity).
We follow Vickers \cite[Section 4 and Section 6]{VickersIntegral} for
the definition of spaces of valuations and covaluations. 
\begin{definition}
A \emph{valuation}
 on a locale $X$ is a Scott continuous
 function $\mu \colon \Frame{X} \to \overrightarrow{[0,\infty]}$ satisfying
\begin{align}
  \notag \mu(0) &= 0, &
  \mu(x) + \mu(y) &= \mu(x \wedge y) + \mu(x \vee y),
\end{align}
where the second condition is called the modular law.
A \emph{covaluation} is a Scott continuous function $\nu \colon \Frame{X} \to
\overleftarrow{[0,\infty]}$ satisfying $\nu(1) = 0$ and
the modular law.
The \emph{space of valuations} $\Valuation{X}$ on a locale $X$ is the
locale whose models are valuations on $X$.
The \emph{space of covaluations} $\CoValuation{X}$ is defined similarly.
\end{definition}

For a strong proximity lattice $(S,\prox)$, the locale
$\Valuation{\Spec{S}}$ is presented by a geometric theory
$T_{\mathfrak{V}}$ over the set
\[
  S_{\mathfrak{V}} \defeql  \left\{ \langle p, a
  \rangle \mid p \in \Rat \amp a \in S \right\}
\]
with the following axioms:
\begin{align*}
    \top &\entails \langle p, a \rangle \quad &&(\text{if $p < 0$}) 
    \\
    \langle p, 0 \rangle &\entails \bot \quad &&(\text{if $0 < p$}) 
    \\
  \langle p, a \rangle
  &\entails \langle q, b \rangle
  &&(\text{if $q \leq p$ and $a \leq b$})\\
  \langle p, a \rangle \wedge  \langle q, b \rangle
  &\entails \bigvee_{p' + q' = p + q}\!\! \langle p', a \wedge b \rangle
  \wedge 
  \langle q', a \vee b \rangle
  \\
  \langle p, a \wedge b \rangle \wedge  \langle q, a \vee b \rangle
  &\entails
  \bigvee_{p' + q' = p + q}\!\! \langle p', a \rangle \wedge \langle q', b \rangle
  \\
  \langle p, a \rangle
  &\entails \langle q, b \rangle
  &&(\text{if $q < p$ and $a \prec b$})\\
  \langle p, a \rangle
  &\entails \bigvee_{p < q, b \prec a}\langle q, b \rangle
\end{align*}
A model $\alpha$ of
$T_{\mathfrak{V}}$ determines a valuation 
$\mu_{\alpha} \colon \RIdeals{S} \to \overrightarrow{[0,\infty]}$ by
\[
  \mu_{\alpha}(I) = \left\{ q \in \Rat \mid  \exists a \in I
  \left(\langle q, a \rangle \in \alpha \right) \right\}.
\]
Conversely, a valuation $\mu \colon \RIdeals{S} \to
\overrightarrow{[0,\infty]}$
determines a model $\alpha_{\mu}$  of
$T_{\mathfrak{V}}$ by 
\[
  \alpha_{\mu} = \left\{ \langle q, a \rangle \in
    S_{\mathfrak{V}}\mid q \in
    \mu(\downset_{\prox} a)\right\}.
\]

The locale $\CoValuation{\Spec{S}}$ 
is presented by a geometric theory $T_{\mathfrak{C}}$ 
obtained from  $T_{\mathfrak{V}}$ by replacing the first three
and the last two axioms with the following:
\begin{equation*}
  \begin{aligned}
   \langle p, a \rangle &\entails \bot \quad &&(\text{if $p < 0$})\\
   \top &\entails \langle p, 1 \rangle \quad &&(\text{if $0 < p$})\\
  \langle p, a \rangle
  &\entails \langle q, b \rangle
  &&(\text{if $p \leq q$ and $a \leq b$})\\
  \langle p, a \rangle
  &\entails \langle q, b \rangle
  &&(\text{if $p < q$ and $a \prec b$})\\
  \langle p, a \rangle
  &\entails \bigvee_{q < p, b \prec a}\langle q, b \rangle
  \end{aligned}
\end{equation*}
It is straightforward to show that $\Valuation{\Spec{S}}$
and $\CoValuation{\Spec{S}}$ 
are isomorphic to the spaces of valuations and covaluations
in Vickers \cite[Section 4 and Section~6]{VickersIntegral}.

The following lemma allows us to present the space of valuations
on $\Spec{S}$ by a strong continuous entailment relation.
\begin{lemma}\label{lem:Valuation}
  Under the other axioms of $T_{\mathfrak{V}}$ (or $T_{\mathfrak{C}}$), the axioms
  \begin{align}
  \label{Ax1}
  \langle p, a \rangle \wedge  \langle q, b \rangle
  &\entails \bigvee_{p' + q' = p + q} \!\!\langle p', a \wedge b \rangle
  \wedge
  \langle q', a \vee b \rangle\\
  \label{Ax2}
  \langle p, a \wedge b \rangle \wedge \langle q, a \vee b \rangle
  &\entails
  \bigvee_{p' + q' = p + q} \!\!\langle p', a \rangle \wedge  \langle q', b
  \rangle 
  \end{align}
  are equivalent to the following axioms:
  \begin{align}
    \langle p, a \rangle \wedge  \langle q, b \rangle
    &\entails 
      \langle r, a \wedge b \rangle \vee \langle s, a \vee b \rangle
    & (\text{if $p + q =  r + s$})\label{Ax2a}\\
    \langle r, a \wedge b \rangle \wedge \langle s, a \vee b \rangle
    &\entails \langle p, a \rangle \vee \langle q, b \rangle
    & (\text{if $p + q =  r + s$})\label{Ax2b}
  \end{align}
\end{lemma}
Here, the equivalence of two axioms means that one axiom holds in the locale
presented by the other axiom and the rest of the axioms of
$T_{\mathfrak{V}}$ (or $T_{\mathfrak{C}}$).

\begin{proof}
The proof is inspired by Coquand and
Spitters~\cite[Lemma 2]{CoquandSpittersIntegralsValuations}, which we
elaborate below.
We identify
generators $S_{\mathfrak{V}}$ with the corresponding elements of 
$\Space{T_{\mathfrak{V}}}$ (or the locale presented by \eqref{Ax2a} and \eqref{Ax2b}
in place of \eqref{Ax1} and \eqref{Ax2}). We write
$\leq_{\mathfrak{V}}$ for the orders in these  locales.

First, assume \eqref{Ax1}. Let $p,q,r,s \in \Rat$ such that
$p + q = r + s$.
  Take any $p',q' \in \Rat$ such that $p' + q' = p + q$.
  If $p' \geq r$, then $\langle p', a \wedge b \rangle 
  \leq_{\mathfrak{V}}
   \langle r, a \wedge b \rangle $. 
  If $p' < r$, then
  $q' = s + (r - p')$, and thus
  $\langle q', a \vee b \rangle \leq_{\mathfrak{V}} \langle s, a \vee b \rangle$.
  Hence
  \[
    \langle p', a \wedge b \rangle \wedge \langle q', a \vee b \rangle
    \leq_{\mathfrak{V}}
    \langle r, a \wedge b \rangle \vee \langle s, a \vee b \rangle
  \]
  for all $p',q' \in \Rat$ such that $p' + q' = p + q$.
  Applying \eqref{Ax1}, we obtain \eqref{Ax2a}.
  Similarly, we obtain  \eqref{Ax2b} from \eqref{Ax2}.

  Conversely, assume \eqref{Ax2a}.
  By the last two axioms of $T_{\mathfrak{V}}$, we have
  \begin{equation}\label{eq:prox}
    \langle q, a \rangle \wedge \langle r, b \rangle
    \leq_{\mathfrak{V}}
    \bigvee_{q' + r' > q +  r} \langle q', a \rangle \wedge
    \langle r', b \rangle.
  \end{equation}
  Let $q', r' \in \Rat$ such that $q' + r' > q + r$. Let 
  $\theta \in \Rat$ such that $q' + r' = q + r +
  \theta$, and choose $N \in
  \Nat$ so large that $q + r + \theta - N \theta < 0$.
  By \eqref{Ax2a}, we have
  \[
    \langle q', a \rangle 
    \wedge \langle r', b \rangle
    \leq_{\mathfrak{V}}
     \langle q + r + \theta - (- \theta + n \theta), a \wedge b \rangle
     \vee
    \langle - \theta + n \theta, a \vee b \rangle
  \]
   for all $n \in \Nat$.
  For each $n \in \Nat$, define
  \[
    \varphi^{n}_{0}
    \defeql \langle q + r + 2  \theta -  n \theta, a \wedge b \rangle,
    \quad
    \qquad
    \varphi^{n}_{1}
    \defeql \langle - \theta + n \theta, a \vee b \rangle.
  \]
  Then, we have
  \begin{equation}\label{eq:phi}
    \langle q', a \rangle \wedge \langle r', b \rangle
    \leq_{\mathfrak{V}}
    \bigwedge_{n \leq N+1}
    \varphi^{n}_{0} \vee \varphi^{n}_{1}
    \leq_{\mathfrak{V}}
    \bigvee_{f \in \mathsf{Ch}(N + 1)} \bigwedge_{n \leq N + 1}
    \varphi^{n}_{f_n},
  \end{equation}
  where $\mathsf{Ch}(N+1)$ is the set 
  of choice functions $f  \colon \left\{ 0,\dots,N+1 \right\} \to \{0,1\}$.
  For each $f  \in \mathsf{Ch}(N+1)$, one of the following cases occurs:
  \smallskip

  \noindent\emph{Case 1}: $\forall n \leq N+1 \, f_{n} = 0$.
  Since $- \theta < 0$,  we have
  \begin{align*}
    \varphi^{0}_{f_{0}}
    \leq_{\mathfrak{V}}
    \langle q + r + \theta, a \wedge b \rangle
    \leq_{\mathfrak{V}}
    \langle q + r + \theta, a \wedge b \rangle \wedge 
    \langle - \theta, a \vee b \rangle.
  \end{align*}

  \noindent\emph{Case 2}: $\forall n \leq N + 1\, f_{n} = 1$.
  Since $q + r + \theta - N \theta < 0$, we have
  \begin{align*}
    \varphi^{N+1}_{f_{N+1}}
    \leq_{\mathfrak{V}}
    \langle q + r - N \theta, a \wedge b \rangle \wedge 
    \langle N \theta, a \vee b \rangle.
  \end{align*}

  \noindent\emph{Case 3}: $\exists n \leq N\, f_{n} = 0 \amp f_{n + 1} = 1$.
  \begin{align*}
    \varphi^{n}_{f_{n}} \wedge \varphi^{n + 1}_{f_{n+1}} 
    &\leq_{\mathfrak{V}}
    \langle q + r - n \theta, a \wedge b \rangle \wedge 
    \langle  n \theta, a \vee b \rangle.
  \end{align*}

  \noindent\emph{Case 4}: $\exists n \leq N\, f_{n} = 1 \amp f_{n + 1} = 0$.
  \begin{align*}
    \varphi^{n}_{f_{n}} \wedge \varphi^{n + 1}_{f_{n+1}} 
    &\leq_{\mathfrak{V}}
    \langle q + r + \theta - n \theta, a \wedge b \rangle \wedge 
    \langle - \theta  + n \theta, a \vee b \rangle.
  \end{align*}
  Thus, in any case
  \[
    \bigwedge_{n \leq N + 1} \varphi^{n}_{f_n}
    \leq_{\mathfrak{V}}
    \bigvee_{q' + r' = q + r} \langle q', a \wedge b \rangle \wedge
    \langle r', a \vee b \rangle.
  \]
  Hence, by \eqref{eq:phi} and \eqref{eq:prox}, we have
  \eqref{Ax1}. Similarly, \eqref{Ax2b} implies \eqref{Ax2}.
\end{proof}

\begin{proposition}\label{prop:ContEntValuation}
  For any strong proximity lattice
  $(S, \prox)$, the locale $\Valuation{\Spec{S}}$
  can be presented by a strong continuous 
  entailment relation 
  $\Valuation{S}
  = {(S_{\mathfrak{V}}, \entails_{\mathfrak{V}}, \prox_{\mathfrak{V}})}$
  where $\entails_{\mathfrak{V}}$ is generated by the following axioms:
  \begin{align*}
    &\entails_{\mathfrak{V}} \langle p, a \rangle \quad &&(\text{if $p <
    0$}) \\
    \langle p, 0 \rangle &\entails_{\mathfrak{V}}  && (\text{if $0 <
    p$})  \\
    \langle p, a \rangle &\entails_{\mathfrak{V}} \langle q, b \rangle && (\text{if $q \leq p$ and $a \leq b$})\\
    \langle p, a \rangle, \langle q, b \rangle
    &\entails_{\mathfrak{V}} \langle r, a \wedge b \rangle, \langle s, a \vee
    b \rangle
    && (\text{if $p + q =  r + s$})\\
    \langle r, a \wedge b \rangle, \langle s, a \vee b \rangle
    &\entails_{\mathfrak{V}}
    \langle p, a \rangle, \langle q, b \rangle
    && (\text{if $p + q =  r + s$})
  \end{align*}
  The idempotent relation $\prox_{\mathfrak{V}}$ is defined by
  \[
    \langle p,a \rangle
    \mathrel{\prox_{\mathfrak{V}}} \langle q, b \rangle \defeqiv q < p \amp a \prox b.
  \]
  The locale $\CoValuation{\Spec{S}}$ can be presented by 
  a strong continuous entailment relation 
  $\CoValuation{S}
  = (S_{\mathfrak{V}}, \entails_{\mathfrak{C}},
  \prox_{\mathfrak{C}})$
  where $\entails_{\mathfrak{C}}$ is generated by the following axioms:
  \begin{align*}
    \langle p, a \rangle &\entails_{\CoValuationFunc} &&(\text{if $p <
    0$}) \\
    &\entails_{\CoValuationFunc} \langle p, 1 \rangle  && (\text{if $0 <
    p$})  \\
    \langle p, a \rangle &\entails_{\CoValuationFunc} \langle q, b \rangle && (\text{if $p \leq q$ and $a \leq b$})\\
    \langle p, a \rangle, \langle q, b \rangle
    &\entails_{\CoValuationFunc} \langle r, a \wedge b \rangle, \langle s, a \vee
    b \rangle
    && (\text{if $p + q =  r + s$})\\
    \langle r, a \wedge b \rangle, \langle s, a \vee b \rangle
    &\entails_{\CoValuationFunc}
    \langle p, a \rangle, \langle q, b \rangle
    && (\text{if $p + q =  r + s$})
  \end{align*}
  The idempotent relation $\prox_{\CoValuationFunc}$ is defined by
  \[
    \langle p,a \rangle
    \mathrel{\prox_{\CoValuationFunc}} \langle q, b \rangle \defeqiv p < q \amp a \prox b.
  \]

In particular, the spaces of valuations and covaluations on a stably
compact locale are stably compact.
\end{proposition}
\begin{proof}
  One can check that
  $\Valuation{S}$
  and
  $\CoValuation{S}$
  satisfy the condition in Lemma \ref{prop:IndGenContEnt}.
  Then, the claim of the proposition follows from Lemma \ref{lem:Valuation}.
\end{proof}
The constructions $\Valuation{S}$ and $\CoValuation{S}$ extend to functors
$\ValuationFunc \colon \SProxLat \to \SContEnt$
and $\CoValuationFunc \colon \SProxLat \to \SContEnt$, which
send each join-preserving proximity relation
$r \colon (S, \prox) \to (S', \prox')$ to join-preserving proximity
maps $\mathfrak{V}(r) \colon \Valuation{S} \to \Valuation{S'}$
and $\mathfrak{C}(r) \colon \CoValuation{S} \to \CoValuation{S'}$
defined by
\begin{align*}
  A \mathrel{\mathfrak{V}(r)} B
  &\defeqiv
   \exists C \in \Fin{S_{\mathfrak{V}}}
  \left( A \entails_{\mathfrak{V}} C \amp
  \forall \langle p, c \rangle \in C  \exists
  \langle q, b \rangle \in B  \left( p > q \amp c \mathrel{r} b\right) \right),\\
  A \mathrel{\mathfrak{C}(r)} B 
  &\defeqiv 
   \exists C \in \Fin{S_{\mathfrak{V}}}
  \left(A \entails_{\mathfrak{C}} C \amp
  \forall \langle p, c \rangle \in C \exists
  \langle q, b \rangle \in B \left(  p < q \amp c \mathrel{r} b\right) \right).
\end{align*}
\begin{theorem}\label{thm:DeGrootDualValuation}
  For any strong proximity lattice $S$,
  we have 
  \[
    \DeGroot{\Valuation{S}} \cong \CoValuation{\DeGroot{S}}
    \quad \text{and} \quad
    \DeGroot{\CoValuation{S}} \cong \Valuation{\DeGroot{S}}.
  \]
\end{theorem}
\begin{proof}
  Immediate from Proposition \ref{prop:ContEntValuation} and
Lemma \ref{lem:GenEntOp}.
\end{proof}

We now focus on \emph{probabilistic}
valuations and covaluations, i.e., those valuations $\mu$ and
covaluations $\nu$ satisfying $\mu(1) = 1$ and $\nu(0) = 1$.

For a strong proximity lattice $(S,\prox)$, the space
$\PValuation{\Spec{S}}$ of probabilistic valuations is presented by
a geometric theory $T_{\mathfrak{V_{P}}}$ which extends the theory
$T_{\mathfrak{V}}$ with the following extra axioms:
\begin{align*}
   \langle p, a \rangle &\entails \bot \quad (\text{if $1 < p$}), &
   \top &\entails \langle p, 1 \rangle \quad (\text{if $p < 1$}).
\end{align*}
The space 
$\PCoValuation{\Spec{S}}$ of probabilistic covaluations is presented by
a geometric theory $T_{\mathfrak{C_{P}}}$ which extends the theory
$T_{\mathfrak{C}}$ with the  following extra axioms:
\begin{align*}
    \top &\entails \langle p, a \rangle \quad (\text{if $1 < p$}),  &
    \langle p, 0 \rangle &\entails \bot \quad (\text{if $p < 1$}).
\end{align*}

Proposition \ref{prop:ContEntValuation}
restricts to probabilistic valuations and covaluations.
\begin{proposition}\label{prop:ContEntCoValuation}
  For any strong proximity lattice
  $(S, \prox)$, the locale $\PValuation{\Spec{S}}$
  can be presented by a strong continuous entailment relation 
  $\PValuation{S}
  = {(S_{\mathfrak{V}}, \entails_{\mathfrak{V_{P}}}, \prox_{\mathfrak{V}})}$
  where $\entails_{\mathfrak{V_{P}}}$ is generated by the axioms of
  $\entails_{\mathfrak{V}}$ and the following extra axioms:
  \begin{align*}
    \langle p, a \rangle &\entails_{\mathfrak{V_{P}}} \quad (\text{if $1
    < p$}), &
    &\entails_{\mathfrak{V_{P}}} \langle p, 1 \rangle \quad (\text{if $p
    < 1$}).
  \end{align*}
  The locale $\PCoValuation{\Spec{S}}$ can be presented by 
  a strong continuous entailment relation 
  $\PCoValuation{S} = (S_{\mathfrak{V}}, \entails_{\mathfrak{C_{P}}},
  \prox_{\mathfrak{C}})$
  where $\entails_{\mathfrak{C_{P}}}$ is generated by the axioms of
  $\entails_{\mathfrak{C}}$ and the following extra axioms:
  \begin{align*}
    &\entails_{\PCoValuationFunc} \langle p, a \rangle \quad (\text{if $1
    < p$}), &
   \langle p, 0 \rangle &\entails_{\PCoValuationFunc}  \quad (\text{if
    $p < 1$}).
  \end{align*}

In particular, the spaces of probabilistic valuations 
and probabilistic covaluations on a stably compact locale are stably
compact.
\qed
\end{proposition}
As a corollary we obtain the probabilistic version of Theorem
\ref{thm:DeGrootDualValuation}.
\begin{theorem}\label{thm:DeGrootDualPValuation}
  For any strong proximity lattice $S$,
  we have 
  \[
    \DeGroot{\PValuation{S}} \cong \PCoValuation{\DeGroot{S}}
    \quad \text{and} \quad
    \DeGroot{\PCoValuation{S}} \cong \PValuation{\DeGroot{S}}.
  \]
\par \vspace{-1.8\baselineskip}
  \qed
\end{theorem}

For probabilistic valuations and covaluations,
we have the following duality.
\begin{lemma}
  \label{lem:ValCovalIso}
  For any strong proximity lattice $S$,  we have
  \[
    A \entails_{\mathfrak{V_{P}}} B 
    \iff
    \DualVal{A} \entails_{\mathfrak{C_{P}}} \DualVal{B} 
  \]
  for all $A, B \in \Fin{S_{\ValuationFunc}}$, where
    $
    \DualVal{A} \defeql \left\{ \langle 1 - p, a \rangle \mid \langle
    p, a \rangle \in A\right\}.
    $
\end{lemma}
\begin{proof}
  The direction $\Rightarrow$ is proved by induction on the derivation of $A \entails_{\mathfrak{V_{P}}} B$.
  Note that each axiom $A\entails_{\mathfrak{V_{P}}} B$ of $\PValuation{S}$
  corresponds to an axiom $\DualVal{A} \entails_{\mathfrak{C_{P}}} \DualVal{B}$ of
  $\PCoValuation{S}$.
  The direction $\Leftarrow$ is similarly proved by induction on $\DualVal{A}
  \entails_{\mathfrak{C_{P}}} \DualVal{B}$.
\end{proof}

Since ``$1$'' in the lower and upper reals form a Dedekind real, 
the following proposition is analogous to Vickers
\cite[Proposition 6.3]{VickersIntegral}, which holds for
an arbitrary locale. We give a proof for the special case of
stably compact locales.

\begin{proposition}\label{prop:ValCovalIso}
  For any strong proximity lattice $S$, we have 
  \[
    \PValuation{S} \cong \PCoValuation{S}.
  \]
\end{proposition}
\begin{proof}
  Define proximity maps 
  $r \colon \PValuation{S} \to \PCoValuation{S}$ and 
  $s \colon \PCoValuation{S} \to \PValuation{S}$ by
    \begin{align*}
      A \mathrel{\mathrel{r}} B &\defeqiv A
      \ll_{\entails_{\PValuationFunc}} \DualVal{B}, &
      B \mathrel{\mathrel{s}} A &\defeqiv B
      \ll_{\entails_{\PCoValuationFunc}} \DualVal{A}.
    \end{align*}
Using Lemma \ref{lem:ValCovalIso}, 
it is straightforward to show that $r$ and $s$ are indeed proximity maps
which are inverse to each other.
\end{proof}

\begin{theorem}\label{thm:ValuationDuality}
  For any strong proximity lattice $S$,
  we have 
  \[
    \DeGroot{\PValuation{S}} \cong \PValuation{\DeGroot{S}}
    \quad \text{and} \quad
    \DeGroot{\PCoValuation{S}} \cong \PCoValuation{\DeGroot{S}}.
  \]
\end{theorem}
\begin{proof}
By Theorem \ref{thm:DeGrootDualPValuation} and Proposition
\ref{prop:ValCovalIso}. 
\end{proof}

Goubault-Larrecq~\cite[Theorem 6.11]{GoubaultLarrecq-ModelofChoice}
proved the corresponding result for stably compact spaces.
Although his proof is classical and the space of covaluations is
implicit in his proof, the essential idea seems to be similar.

\subsection*{Acknowledgements}
I thank the referees for numerous suggestions which help me 
improve the paper in an essential way. In particular,
their suggestion to use entailment systems allows me to simplify
and constructivise many parts of the paper. I also thank Steve
Vickers and Daniel Wessel for helpful discussions. This work was
carried out while I was in the Hausdorff Research Institute for
Mathematics (HIM), University of Bonn, for their trimester program
``Types, Sets and Constructions'' (May--August 2018). I thank the
institute for their support and the organisers of the program for
creating a stimulating environment for research.

\end{document}